%% file: main.tex
\documentclass[conference]{IEEEtran}
\IEEEoverridecommandlockouts
\usepackage{cite}
\usepackage{amsmath,amssymb,amsfonts}
\usepackage{graphicx}
\usepackage{textcomp}
\usepackage{xcolor}
\usepackage{lineno}
\modulolinenumbers[5]
\usepackage[noend]{algpseudocode}
\usepackage[linesnumbered,ruled,vlined]{algorithm2e}
\usepackage{caption}
\usepackage[singlelinecheck=on]{caption}
\usepackage{amsmath,amssymb,amsfonts}
\usepackage{amsthm}
\usepackage{booktabs}
\usepackage{float} 
\usepackage{flushend}
\usepackage{graphicx}
\usepackage{mathrsfs}
\usepackage{enumitem}
\usepackage{multirow}
\usepackage{accents}
\usepackage{ragged2e}
\usepackage{setspace}
\usepackage{subfigure}
\usepackage{tabularx}
\usepackage{textcomp}
\usepackage{url}
\usepackage{verbatim}
\usepackage{xcolor}
\usepackage{xspace}
\usepackage{cite}
\usepackage[normalem]{ulem}
\usepackage[flushleft]{threeparttable}

\input{cmd}

\def\BibTeX{{\rm B\kern-.05em{\sc i\kern-.025em b}\kern-.08em
    T\kern-.1667em\lower.7ex\hbox{E}\kern-.125emX}}

\SetKwProg{proc}{Procedure}{}{}
\newtheorem{example}{Example}
\newtheorem{theorem}{Theorem}
\newtheorem{lemma}{Lemma}

\newtheorem{definition}{Definition}

\renewcommand{\baselinestretch}{0.955}

\begin{document}

\title{Accelerating K-Core Computation in Temporal Graphs
}

\author{\IEEEauthorblockN{Zhuo Ma}
\IEEEauthorblockA{
\textit{University of New South Wales}\\
Sydney, Australia \\
zhuo.ma@student.unsw.edu.au}
\and
\IEEEauthorblockN{Dong Wen}
\IEEEauthorblockA{
\textit{University of New South Wales}\\
Sydney, Australia \\
dong.wen@unsw.edu.au}
\and
\IEEEauthorblockN{Hanchen Wang}
\IEEEauthorblockA{
\textit{AAII, University of Technology Sydney}\\
Sydney, Australia \\
hanchen.wang@uts.edu.au}
\\
\and
\IEEEauthorblockN{Wentao Li}
\IEEEauthorblockA{
\textit{University of Leicester}\\
Leicester, United Kingdom \\
wl226@leicester.ac.uk}
\and
\IEEEauthorblockN{Wenjie Zhang}
\IEEEauthorblockA{
\textit{University of New South Wales}\\
Sydney, Australia \\
wenjie.zhang@unsw.edu.au}
\and
\IEEEauthorblockN{Xuemin Lin}
\IEEEauthorblockA{
\textit{Shanghai Jiaotong University}\\
Shanghai, China \\
xuemin.lin@gmail.com}
}

\maketitle

\begin{abstract}
We address the problem of enumerating all temporal $k$-cores given a query time range and a temporal graph, which suffers from poor efficiency and scalability in the state-of-the-art solution.
% Existing solutions suffer from poor scalability, either due to the necessity of maintaining cohesive sub-graphs for each time window or the requirement to manage previously induced results to ensure output distinctness. 
Motivated by an existing concept called core times, we propose a novel algorithm to compute all temporal $k$-cores based on core times and prove that the algorithmic running time is bounded by the size of all resulting temporal $k$-cores, which is optimal in this scenario. Meanwhile, we show that the cost of computing core times is much lower, which demonstrates the close relevance between our overall running time and the result size.
% By comparison, the time to compute core times is much shorter.
%
% We also study a variation of the existing temporal $k$-core model to avoid the output redundancy and extend the algorithm for the new model.
%
We conduct extensive experiments to demonstrate the efficiency of our proposed method and the significant improvement over existing solutions.
\end{abstract}


\input{1intro}
\input{2prelim} % change accordingly Efficient $(\alpha,\beta)$-core Computation
\input{3existing} % Exact and naive algorithm 
\input{4ct}
\input{5enum}

\input{6experiment}

\section{Conclusion}
We proposed an efficient framework for enumerating temporal $k$-cores in a given time range of temporal graphs. By utilizing core times and introducing minimal core windows for edges, our algorithm achieves enumeration time proportional to the result size, optimal in practice. Experimental results show our approach outperforms existing methods by up to two orders of magnitude across diverse datasets.

% In this work, we proposed an efficient framework for enumerating all temporal $k$-cores over a given time range in temporal graphs. By leveraging the concept of core times and introducing a minimal core window structure for edges, our temporal $k$-cores enumeration algorithm bounds the enumeration time by the size of the resulting temporal k-cores, which is optimal in practice. Extensive experimental evaluations demonstrated that our algorithm consistently outperforms existing methods by up to two orders of magnitude across diverse datasets.

%Additionally, we proposed the temporal $k$-core vertex set model, which offers the same informational value for $k$-core based cohesive sub-graph mining applications but with significantly reduced computational costs compared to existing temporal $k$-core models. 

\stitle{Future Works} The primary goal of this paper is to improve the efficiency of existing algorithms for enumerating temporal $k$-core edge sets, ensuring scalability for large temporal graphs. While our focus has been on edge set enumeration, we acknowledge the potential limitations of this approach in real-world applications. Representing $k$-cores as distinct vertex sets may be more practical and efficient, especially when edge set size poses computational challenges. In future work, we plan to develop techniques for enumerating temporal $k$-core vertex sets, offering a more compact representation and addressing vertex set overlap efficiently. We believe this direction will further enhance the practicality and applicability of temporal $k$-core analysis.

% The primary goal of this paper is to improve the efficiency of existing algorithms for enumerating temporal $k$-core edge sets, ensuring scalability for large temporal graphs. While our focus has been on edge set enumeration, we acknowledge the potential limitations of this approach in real-world applications. Representing $k$-cores as distinct vertex sets may often be more practical and efficient, particularly in scenarios where the edge set size introduces computational challenges. As part of future work, we plan to extend our study to include techniques for enumerating temporal $k$-core vertex sets, offering a more compact and scalable representation. This extension will require additional algorithmic innovations to address vertex set overlap efficiently while maintaining high performance. We believe this direction will further enhance the practical relevance and applicability of temporal $k$-core analysis.

\newpage
\bibliographystyle{IEEEtran}
\bibliography{sample}
\end{document}

%% file: cmd.tex
\newenvironment{remark}{{\vspace{0.5em} \noindent\it Remark.}}

\newcommand{\kw}[1]{{\ensuremath {\mathsf{#1}}}\xspace}

\newcommand{\reffig}[1]{Figure~\ref{fig:#1}}

\newcommand{\refsubsec}[1]{Section~\ref{subsec:#1}}
\newcommand{\reftab}[1]{Table~\ref{tab:#1}}
\newcommand{\refalg}[1]{Algorithm~\ref{alg:#1}}
\newcommand{\refdef}[1]{Definition~\ref{def:#1}}
\newcommand{\refthm}[1]{Theorem~\ref{thm:#1}}
\newcommand{\reflem}[1]{Lemma~\ref{lem:#1}}

\newcommand{\stitle}[1]{\noindent{\bf #1}}

\newcommand{\nbr}{N}

\newcommand{\avgdegree}{deg_{avg}}

\newcommand{\tmax}{t_{max}}
\newcommand{\gwindow}{\mathcal{W}}

\newcommand{\pair}[1]{\langle #1 \rangle}

\newcommand{\vctidx}{\mathcal{VCT}}
\newcommand{\yestime}{active}

\newcommand{\rset}{\mathcal{R}}

\newcommand{\kwnull}{\ensuremath{\kw{Null}}}
\newcommand{\kwcontinue}{\textbf{continue}}

\newcommand{\kwfalse}{\ensuremath{\kw{False}}}
\newcommand{\kwtrue}{\ensuremath{\kw{True}}}
\newcommand{\dll}{\ensuremath{\mathcal{L}}}

\newcommand{\algdlldelete}{\kw{Delete}}
\newcommand{\algdllinsert}{\kw{Insert}}

\newcommand{\edgeskyline}{\mathcal{ECS}}

\newcommand{\algbaseenum}{\kw{EnumBase}}

\newcommand{\algstartanchor}{\kw{AS\text{-}Output}}

\newcommand{\algenum}{\kw{Enum}}

\newcommand{\ct}{\ensuremath{\mathcal{CT}}}

\newcommand{\algct}{\ensuremath{\kw{CoreTime}}}

%% file: 1intro.tex
%TEX root = main.tex
\section{Introduction}
\label{sec:intro}

The temporal graph, where each edge is associated with a timestamp, models various real-world interactions between entities, such as bank transactions and social network interactions over time. For example, in a money transaction network, vertices represent bank accounts and temporal edges represent transactions between accounts at specific times. 
The $k$-core model \cite{Seidman1983} is a fundamental concept to identify cohesive subgraphs, drawing significant research attention due to its wide range of applications, including community detection, network visualization, and system structure analysis \cite{Cui2014,Li2018,Zhang2010,Cheng2011,Khaouid2005,Montresor2011,Wen2015}. Given a graph, a $k$-core is defined as a maximal induced subgraph where each vertex has a degree of at least $k$. 

Identifying $k$-cores in temporal graphs is valuable for tasks such as detecting suspicious account networks during anti-money laundering efforts within specific time frames \cite{Chu2019, Lin2024, Starnini2021}.
To capture the $k$-core structure in temporal graphs, Yang et al. \cite{Yang2023} study the time-range $k$-core query, which enables detecting cohesive subgraphs in a flexible time range. Specifically, given a temporal graph, a query integer $k$, and a query time range, the problem aims to enumerate $k$-cores appearing in the snapshot over any time windows within the query time range. The snapshot over a time window is an unlabeled graph induced by all edges falling in the time window. They use the term temporal $k$-core for a time window to distinguish it from the $k$-core model in unlabeled graphs.

\begin{example}
Consider the temporal graph $G$ in \reffig{example}. Given a query time range $[1,4]$ and $k=2$, we have two temporal 2-cores as shown in Figure \ref{fig:2cores}. For the temporal $k$-core containing $\{v_1,v_2,v_4\}$, it is the $k$-core in the snapshots over the windows $[1,3], [2,3], [2,4]$. 
\end{example}
\vspace{-0.5em}

The time-range $k$-core query generalizes the historical $k$-core query \cite{Yu2021}, which focuses on a single time window, to address the need for analyzing temporal graphs across a continuous range of time windows. While certain applications may only require the $k$-core of a specific time window, many real-world problems demand a more exhaustive exploration of all possible sub-windows. For example, in social network analysis for misinformation detection \cite{Manurung2023, Oettershagen2023, Liu2024}, coordinated misinformation campaigns often unfold in bursts over varying time scales. Identifying the exact temporal pattern of such campaigns requires analyzing multiple overlapping time windows to detect tightly connected groups of users whose activity might not align with any pre-defined windows. By enumerating all temporal $k$-cores, platforms can uncover nuanced patterns, such as recurring bot activity or troll farms, that would be overlooked by static or single-window analysis. For instance, consider a scenario during a national election where misinformation is amplified by bots over short time windows. A single-window query might miss detecting these bursts, but an exhaustive enumeration ensures that all activity windows, regardless of their duration, are examined. Similarly, in disease outbreak monitoring and contact tracing \cite{Nguyen2021, Serafino2021}, transmission clusters may emerge and dissipate rapidly over short and irregular timeframes. Enumerating all $k$-cores across a range of time windows ensures that even fleeting, high-risk clusters are identified, enabling proactive measures like targeted quarantines or containment strategies. For instance, during a sudden outbreak in a community, interactions between infected individuals may peak and decline over unpredictable durations. An exhaustive temporal $k$-core enumeration enables health authorities to reconstruct transmission chains more accurately, even when precise timeframes are initially unknown.

\begin{figure}[tbp]
% \vspace{10pt}
\centering
\includegraphics[scale=0.31]{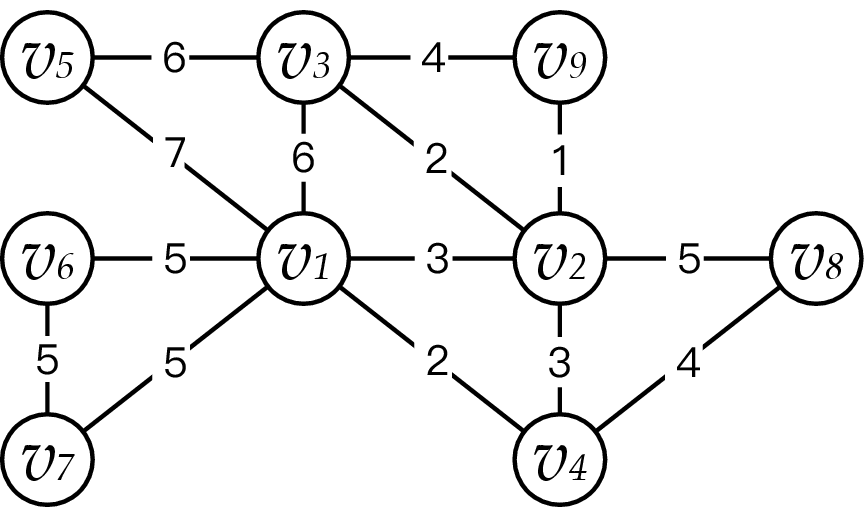}
\vspace{-0.5em}
\caption{A temporal graph $G$.}
\label{fig:example}
\vspace{0.5em}
\end{figure}

\begin{figure}[tbp]
\centering
\includegraphics[scale=0.31]{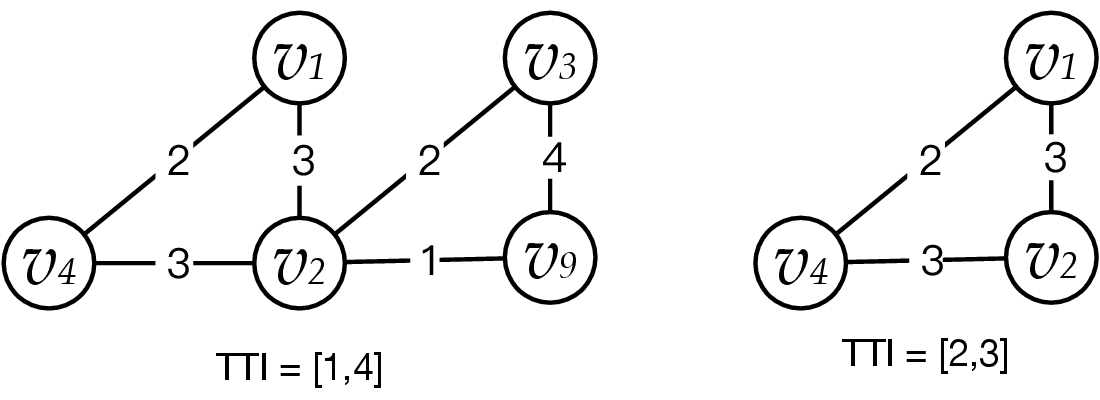}
\vspace{-0.5em}
\caption{The temporal $2$-cores of the graph in \reffig{example} given the query time range $[1,4]$.}
\label{fig:2cores}
\vspace{0.5em}
\end{figure}

\stitle{The State of the Art.} To enumerate all temporal $k$-cores in a time range, the general idea of \cite{Yang2023} is to iterate all time windows within the time range. Note that the $k$-core of a simple unlabeled graph can be computed by a peeling algorithm that continuously removes vertices with fewer than $k$ neighbors. Motivated by this, they first process the widest window and iteratively narrow the window. Benefiting from this strategy, the temporal $k$-core of a window can be directly computed from the result of the previous window. Certain data structures and optimizations are also designed to speed up the algorithm. For instance, they skip a window if they already find the temporal $k$-core of the current window is the same as that of the previous window based on several properties.
Despite a spectrum of optimizations, there is no clear better time complexity than $O(\tmax^2 * B)$ shown in \cite{Yang2023}, where $\tmax$ is the number of timestamps in the query time range and $B$ is the average cost of computing a temporal $k$-core for one time window. The main bottleneck of their method lies in $O(\tmax^2)$ iterations to scan windows. Note that $\tmax$ is up to the number of edges in the query range if the time label of each edge is unique. 
% In addition, \rr{their pruning strategies cannot exactly filter out the window that will produce the duplicated temporal $k$-core,} and that implies much computational cost in computing the same result.

% To enumerate all temporal $k$-cores in a time range, \cite{Yang2023} iterates over all time windows within the range. The $k$-core of a simple graph is computed using a peeling algorithm that removes vertices with fewer than $k$ neighbors. Starting with the widest window, their approach narrows it iteratively, leveraging results from previous windows to compute the $k$-core of the current one. Optimizations, such as skipping windows where the $k$-core remains unchanged, improve efficiency. However, their method has a time complexity of $O(\tmax^2 * B)$ shown in \cite{Yang2023}, where $\tmax$ is the number of timestamps in the query time range and $B$ is the average cost of computing a temporal $k$-core for one time window. The bottleneck lies in the $O(\tmax^2)$ iterations to scan windows, especially when $\tmax$ approaches the number of edges if each edge has a unique timestamp.

\stitle{The framework.} In this paper, we propose a novel algorithm to significantly speed up the time-range $k$-core query. Our idea is inspired by a concept named core time \cite{Yu2021}, which represents the earliest end time $te$ for a vertex $u$ and a start time $ts$ such that $u$ is in the $k$-core of the snapshot of the time window $[ts,te]$. The core time is a key component of the PHC index proposed in \cite{Yu2021}, which efficiently computes $k$-cores across all $k$-values for a given temporal graph. In this work, we focus on a specific $k$ and use only the portion of the PHC index relevant to that $k$, which we refer to as the Vertex Core Time index ($\vctidx$) for clarity. To compute $\vctidx$, \cite{Yu2021} proposes an algorithm with time complexity of $O(|\vctidx|*\avgdegree)$ where $|\vctidx|$ is the size of the vertex core time index and $\avgdegree$ is the average degree in the query graph. We observe that the algorithm can be easily extended to compute minimal core windows of all edges as a byproduct with the same theoretical running time. A minimal core window of an edge $e$ is a minimal time window $[ts,te]$ such that $e$ is in the $k$-core of the snapshot over $[ts,te]$. Intuitively, the minimal core windows of edges compress the relationship between the edge and the $k$-cores of all possible time windows. Therefore, we expect to have an efficient algorithm to directly compute all temporal $k$-cores based on the minimal core windows, and the method immediately avoids computing $k$-cores of individual time windows by analyzing the relationship with minimal core windows and resulting subgraphs.

While the historical $k$-core problem and its PHC index focus on efficiently computing $k$-cores for individual time windows, the time-range $k$-core problem introduces unique challenges. Specifically, the goal is to enumerate all $k$-cores across all overlapping time windows within a query time range. This task involves avoiding redundancy, as the same $k$-core may appear in multiple overlapping time windows. Enumerating all temporal $k$-cores thus requires careful handling of overlaps to ensure that each $k$-core is represented only once, significantly increasing the complexity compared to historical $k$-core computations.

% In this paper, we propose a novel algorithm to accelerate time-range $k$-core queries. Inspired by the core time concept [26], which identifies the earliest end time $te$ for a vertex $u$ in the $k$-core of a snapshot $[ts,te]$, we focus on a specific $k$ and introduce the Vertex Core Time index ($\vctidx$) as a subset of the PHC index \cite{Yu2021}. $\vctidx$ computation has a time complexity of $O(|\vctidx|*\avgdegree)$, where $|\vctidx|$ is the size of the vertex core time index and $\avgdegree$ is the average degree. We extend this approach to compute minimal core windows for edges, compactly encoding their relationships with $k$-cores across all possible time windows. These windows enable efficient enumeration of temporal $k$-cores, eliminating the need to compute $k$-cores for individual time windows and avoiding redundancy. 

% Unlike the historical $k$-core problem, which focuses on individual time windows or specific $k$-values, the time-range $k$-core problem involves enumerating all $k$-cores across overlapping time windows. This introduces significant challenges, as overlapping windows can lead to redundancy if not handled carefully. Our approach ensures that each $k$-core is represented only once, addressing these overlaps and adding complexity compared to historical $k$-core computations.

\stitle{Efficient Enumeration.} Our main technical contribution is an efficient method to enumerate all temporal $k$-cores using the minimal core windows of edges. Starting from a given start time $ts$, we filter out minimal core windows that do not contribute to any $k$-core at $ts$ and sort the remaining windows by end times. We iteratively scan these windows, outputting the corresponding edges of all previous windows as a $k$-core when specific conditions are met, ensuring no duplicate results. To update the data structure from one start time to the next, we use a doubly linked list to maintain window order. By preprocessing all minimal core windows in linear time, we can update the order from $\dll'$ to $\dll$ in $O(|\dll \setminus \dll'|)$ time complexity, where $\dll \setminus \dll'$ is the difference between the set of windows in two orders. Consequently, the time complexity for enumerating $k$-cores for all start times is bounded by the result size, making it optimal. Overall, our solution has a time complexity of $O(|\vctidx|*\avgdegree+|\rset|)$, where $|\vctidx|*\avgdegree$ accounts for computing minimal core windows and $|\rset|$ represents the result size. We report the $|\vctidx|*\avgdegree$ and $|\rset|$ for several representative real datasets in \reffig{ct_vs_ressize}. The result size is much larger than $|\vctidx|*\avgdegree$ in all datasets, which demonstrates that the overall running time of our algorithm is mainly related to the result size.

\stitle{Contribution.} We summarize our main contributions below.

\begin{itemize}
\item\textit{A Novel Framework for Time-Range $K$-Core Queries.} We propose a new framework to precompute the minimal core windows of all edges instead of directly computing $k$-cores of each individual window in the state-of-the-art algorithm. \refsubsec{naive} proposes a basic implementation to enumerate all results based on minimal core windows.
% Our experiments show that the final algorithm is already much more efficient than the state-of-the-art algorithm.
\item\textit{Enumeration in Optimal Time Complexity.} Given all minimal core windows of edges, we propose an algorithm to enumerate all temporal $k$-cores in the time bounded by the result size. %The algorithm is also extended to enumerate vertex sets of temporal $k$-cores (a variation of the existing problem) in optimal time complexity.
% \item\textit{An Improved Model to Reduce Redundancy.} Two temporal $k$-cores are outputted even though they have the same set of vertices but the difference of one edge in the time-range $k$-core query. To reduce redundant results, we extend our technique to a new problem that enumerates the distinguished vertex set of any temporal $k$-core. By applying certain non-trivial strategies, the overall running time of the new problem can still be bounded by $O(|\vctidx| \cdot \avgdegree+|\rset_{V}|)$ where $|\rset_{V}|$ is the result size of the new problem.
\item\textit{Extensive Performance Studies.} In our experiments, our final algorithm is two orders of magnitude faster than the state-of-the-art algorithm for the same problem on most datasets.
\end{itemize}

%% file: 2prelim.tex
\vspace*{-0.5em}
\section{Preliminary}
\label{sec:pre}

\begin{figure*}[t!]
\centering
\includegraphics[scale=0.4]{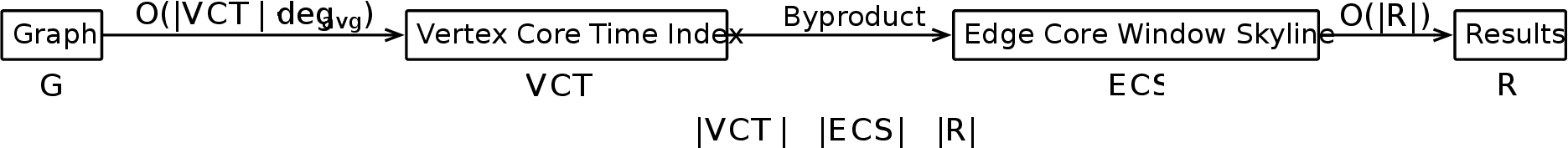}
\caption{Overview of our framework.}
\label{fig:overview}
\end{figure*}

We study an undirected temporal graph $G(V,E)$ where each edge $(u,v,t) \in E$ is associated with a timestamp $t$ that indicates the interaction time between vertices $u$ and $v$. We use $E_t$ to denote all edges with the associated time $t$. 
% We denote the number of vertices as $n$ and the number of edges as $m$. 
Without loss of generality, we denote the timestamps of edges as a continuous set of integers starting from 1. We assume each pair of vertices has at most one edge for simplicity, and our solution can be easily extended for the existence of multiple edges between two vertices. We use $deg(u)$ to denote the degree of a vertex $u$. The projected graph of $G$ over a time window $[ts,te]$, denoted by $G_{[ts,te]}$ is the temporal subgraph including all edges in $[ts,te]$.
% Given a temporal graph $G$, we say $G'(V',E')$ is an induced temporal subgraph of $G$ if $V' \subseteq V$ and $E' = \{(u,v,t) \in E | u \in V', v \in V', t \in \gwindow(G')\}$. The life window of a temporal subgraph $G'(V',E')$, denoted by $\gwindow(G')$, is the minimal window that contains all edges in $E'$, i.e., $\gwindow(G')=[\min_{(u,v,t) \in E'}t,\max_{(u,v,t) \in E'}t]$. 

\begin{definition}[$K$-Core \cite{Seidman1983}]
\label{def:k-core}
Given a simple graph $G$ and an integer $k$, the $k$-core of $G$ is the maximal induced subgraph of $G$ in which every vertex has at least $k$ neighbors.
\end{definition}

% Given a temporal graph $G$ and a time window $[ts,te]$, we consider $G_{[ts,te]}$ as a simple graph by ignoring all edge times when mentioning the $k$-core of $G_{[ts,te]}$. 
The concept of temporal $k$-core extends the $k$-core model for temporal graphs. We present it as follows.

\begin{definition}[Temporal $K$-Core \cite{Yang2023}]
\label{def:temporal-core}
Given a temporal graph $G$ and a window $[ts,te]$, a temporal $k$-core of $[ts,te]$ is a maximal subgraph $C$ of $G_{[ts,te]}$ where every vertex has at least $k$ neighbors.
% $V_{[ts,te]} = \{u \in V | deg(u) \geq k\}$ and $E_{[ts,te]} = \{(u,v,t) \in E | deg(u) \geq k, deg(v) \geq k, ts \leq t \leq te\}$.
\end{definition}

% \begin{definition}[Distinctness]
% \label{def:distinctness}
% Given a temporal graph $G$ and a window $[Ts,Te]$, a temporal $k$-core $\eset_{[ts,te]}$, where $[ts,te] \subseteq [Ts,Te]$, is considered to be distinct within $[Ts,Te]$ if there does not exist another temporal $k$-core $\eset'_{[ts',te']}$ such that:
% \begin{enumerate}
%     \item $\eset_{[ts,te]} = \eset'_{[ts',te']}$, and
%     \item $[ts',te'] \subseteq [ts,te]$
% \end{enumerate}
% \end{definition}

To study $k$-cores in a query time range of a temporal graph, we present the research problem studied in \cite{Yang2023} as follows.

\stitle{Problem Statement.} Given a temporal graph $G$, integer $k$, and a time range $[Ts,Te]$, we aim to compute temporal $k$-cores of all time windows $[ts,te] \subseteq [Ts,Te]$.

% \begin{figure}[htbp]
%     \centering
%     \begin{subfigure}[b]{0.2\textwidth}
%         \centering
%         \includegraphics[width=\textwidth]{figs/example_1-4.eps}
%         \caption{The temporal $k$-cores of $G_{[1,4]}$.}
%         \label{fig:sub1}
%     \end{subfigure}
%     \hfill
%     \begin{subfigure}[b]{0.1\textwidth}
%         \centering
%         \includegraphics[width=\textwidth]{figs/example4-6.eps}
%         \caption{The temporal $k$-core of $G_{[4,6]}$.}
%         \label{fig:sub2}
%     \end{subfigure}
% \end{figure}

% \begin{figure}[tbp]
% \centering
%     \subfigure[Temporal $2$-cores for $G_{[1,4]}$.]{
%         \label{fig:ex14}
%         \includegraphics[width=30mm]{figs/figure2-a.eps}
%     }
%     \subfigure[Temporal $2$-cores for $G_{[3,5]}$.]{
%         \label{fig:ex35}
%         \includegraphics[width=11mm]{figs/figure2-b.eps}
%     }
% \caption{Examples of temporal 2-cores. \rr{need update}}
% \end{figure}

The results for the query window $[1,4]$ and the query integer $k=2$ on the graph $G$ of \reffig{example} is shown in \reffig{2cores}.
Given a query time range $[Ts,Te]$, we use $\nbr(u)$ to denote the neighbors of $u$ in $G_{[Ts,Te]}$ for ease of presentation when the context is clear. 
% Without loss of generality, we denote times of edges as a continuous set of integers starting from $1$ to $\tmax$.  
We use $\tmax$ to denote the number of distinct time labels in the query range, i.e., $\tmax = Te-Ts+1$. We use $n$ and $m$ to denote the number of vertices and the number of edges in the projected graph over the query range $[Ts,Te]$.
Two temporal $k$-cores are considered the same in \cite{Yang2023} if they have the same set of edges. The same temporal $k$-core may exist in multiple time windows, and any solution for the problem should avoid repeated outputs. 
We omit proofs for some lemmas and theorems in the paper if they are straightforward.

%% file: 3existing.tex
%TEX root = main.tex

\section{Existing Studies}

\begin{algorithm}[t]
\caption{OTCD}
\label{alg:otcd}
\KwIn{$G$, $k$, $[Ts, Te]$}
\KwOut{all distinct temporal $k$-cores in $[Ts,Te]$}
\For{$ts \gets Ts$ to $Te$}{
    $te \gets Te$\;
    \If {$ts = Ts$} {
        obtain $C_ {[ts,te]}$ by truncating $G_{[Ts,Te]}$ and performing core decomposition on it\;
    }
    \Else {
        obtain $C_{[ts,te]}$ by truncating $C_{[ts-1,te]}$ and performing core decomposition on it\;
    }
    output $C_{[ts,te]}$ if it is not marked as pruned\;
    \For{$te \gets Te - 1$ to $ts$} {
        \If{$[ts, te]$ is not marked as pruned} {
            obtain $C_{[ts,te]}$ by truncating $C_{[ts,te+1]}$ and performing core decomposition on it\;
            compute time windows that need to be pruned based on TTI of the $C_{[ts,te]}$\;
            output $C_{[ts,te]}$\;
        }
    }
}
\end{algorithm}

\subsection{The State of the Art}

To enumerate temporal $k$-cores, \cite{Yang2023} proposed an algorithm called Optimized Temporal Core Decomposition (OTCD). The algorithm computes temporal $k$-cores decrementally by considering the projected graph from wide time windows to narrow time windows. \refalg{otcd} presents the algorithm. Given a query time range $[Ts, Te]$, OTCD first removes all edges not in $[Ts, Te]$ and computes the $k$-core in the truncated graph by iteratively removing all vertices that do not have at least $k$ neighbors. Next, OTCD enumerates each window $[ts, te] \subseteq [Ts, Te]$ in a specific order to compute the temporal $k$-core for the corresponding time windows decrementally from previously computed temporal $k$-cores. Specifically, OTCD initializes with $ts = Ts$ and $te = Te$. Then, it anchors the start time $ts$ and decreases the end time $te$ from $Te$ to $ts$. Once $te$ reaches $ts$, the algorithm increments $ts$ by one and repeats the process until $ts$ reaches $Te$. OTCD optimizes the algorithm by pruning some time windows that do not have any unique temporal $k$-cores. To this end, they define the following concept called Tightest Time Interval.

\begin{definition}[Tightest Time Interval]
\label{def:tti}
The tightest time interval (TTI) of a temporal $k$-core $C$, denoted by $\gwindow(C)$, is the minimal time window containing all edges in $C$.
\end{definition}

Based on \refdef{tti}, OTCD applies three rules: Pruning-on-the-Right (PoR), Pruning-on-the-Underside (PoU), and Pruning-on-the-Left (PoL) to avoid computing temporal $k$-cores of unnecessary time windows. For a start time $ts$ and an end time $te$, a temporal $k$-core is derived with the TTI $[ts',te']$. PoR prunes all the time windows starting from $ts$ ending from $te'$ to $te$. In the case of $ts' > ts$, PoU prunes the all time windows starting not earlier than $ts$ and ending later than $ts'$. 
% Both PoR and PoU prune windows that induce the repeated temporal $k$-cores.
%
If we have both $ts' > ts$ and $te' < te$, in addition to PoR and PoU, PoL prunes time windows starting after $ts'$ ending in the range $[te'+1,te]$.
Certain data structure is designed in \cite{Yang2023} to implement the pruning rules efficiently.

% For example, given a time window $I'$, its next window $I$ is pruned if $I \subset I'$ and the TTIs of all $k$-cores computed by $I'$ is contained in $I$.
%
% \rr{OTCD prunes future time windows based on the relationship between the current time window and the TTI of its induced temporal $k$-core. Depending on variations in the start time, end time, or both from TTI, three kinds of pruning techniques are triggered in different circumstances. Yang et al. (2023) conclude that there is a one-to-one correspondence between temporal $k$-cores and TTIs, allowing OTCD to eliminate redundant computations by identifying temporal $k$-cores that are identical to previously induced ones.}
%

\stitle{Challenges.} 
Despite practical optimizations to prune certain time windows, the OTCD algorithm essentially checks every windows in the query range and is time-consuming. The time complexity of OTCD is $O(|\tmax| * (m * log(n) + m))$, where $O(m * log(n) + m)$ is the running time to compute temporal $k$-cores for ranges from one start time and all possible end times. The time complexity can also be represented as $O(|\tmax|^2 * B)$, where $B$ represents the average cost of computing the temporal $k$-core for one time window.

\subsection{Other Related Works}
\label{subsec:other}

Various $k$-core-related problems have been studied in temporal graphs, incorporating different temporal objectives and constraints beyond cohesiveness. Historical $k$-cores focus on snapshots at specific times \cite{Yu2021}, while maximal span-cores require edges to appear continuously throughout a time range \cite{Galimberti2018}. The $(\pi, \tau)$-persistent $k$-core maintains a $k$-core across $\tau$-length sub-windows if its persistence exceeds $\tau$ \cite{Li2018}. Continual cohesive subgraph search finds subgraphs containing a queried vertex under structural and temporal constraints \cite{Li2021}. Dense subgraphs in weighted temporal networks explore fixed vertex sets with varying edge weights \cite{Ma2020}. The $(k, h)$-core ensures each vertex has at least $k$ neighbors with $h$ interactions \cite{Wu2015}, while density bursting subgraphs (DBS) identify subgraphs with the fastest-growing density over time \cite{Chu2019}. Periodic community detection reveals recurring interaction patterns such as periodic $k$-cores \cite{Qin2019, Qin2022}. Finally, quasi-$(k, h)$-cores provide detailed measures for temporal graphs and focus on efficient maintenance \cite{Bai2020}. These models enhance the understanding of cohesive substructures in dynamic networks.

%% file: 4ct.tex
%TEX root = main.tex
\section{Solution Overview}
\label{sec:overview}

% \begin{figure*}[tbp]
% \vspace{10pt}
% \centering
% \includegraphics[scale=0.4]{figs/ct.eps}
% \caption{$\ctidx$ and $\window$ for the temporal graph in Figure \ref{fig:example} for $k = 2$.}
% \label{fig:ct}
% \end{figure*}

\subsection{The Framework}

% \begin{algorithm}[t!]
% \caption{Our Framework}
% \label{alg:frk}
% \KwIn{a temporal graph $G$, an integer $k$, a time range $[Ts,Te]$}
% \KwOut{all temporal $k$-cores in $G_{[Ts,Te]}$}
% \tcc{Compute all minimal core windows of each edge}
% $\edgeskyline \gets \algct(G,Ts,Te,k)$\;
% \tcc{Enumerating all temporal k-cores}
% $\algedgeenum(\edgeskyline,Ts,Te)$\;
% % \vspace{10pt}
% \end{algorithm}

To improve the efficiency of temporal $k$-core enumeration, we present an overview of our solution in \reffig{overview}. Our idea is inspired by a study on computing vertex core times \cite{Yu2021}, which is formally defined as follows.

\begin{definition}[Vertex Core Time]
\label{def:vct}
Given a temporal graph $G$, an integer $k$, a start time $ts$ and a vertex $u$, the core time of $u$, denoted by $\ct_{ts}(u)_k$, is the earliest end time $te$ such that $u$ is in the temporal $k$-core of $G_{[ts,te]}$.
\end{definition}

\begin{example}
Refer to the temporal graph in \reffig{example}. For $k = 2$, we compute the core time of vertices at the start time $ts = 1$ by considering time windows starting at $[ts,te] = [1,1]$ and increasing the end time $te$ from 1 to the largest timestamp 7. As we expand the time window to $[1,3]$, $v_1$ joins the 2-core. Therefore, its core time for $ts = 1$ is 3, i.e. $\ct_1(v_1)_2 = 3$. Similarly, for $ts = 3$, $\ct_3(v_1)_2 = 5$.
\end{example}

We omit the subscript $k$ in \refdef{vct} when it is clear from the context. By replacing vertex with edge, we also have the concept of edge core time. The existing work \cite{Yu2021} computes the core times of all vertices for all start times.
%  for all possible integer $k$, and we only focuses on one specific integer $k$. 
Given that the core time of a vertex for a series of continuous start times may be the same, they only output distinct core times with the corresponding earliest start times. For simplicity, we call this structure Vertex Core Time Index, denoted by $\vctidx$.
%
% \reftab{vct} presents the vertex core time index for the temporal graph in \reffig{example}.

\begin{table}[tbp]% h asks to places the floating element [h]ere.
    \centering
    \caption{The vertex core time index of the temporal graph $G$ \reffig{example} for $k=2$.}
    \label{tab:vct}
  \begin{tabular}{ll}
    \toprule
    $v_1$: $[1,3], [3,5], [6,7], [7,\infty]$ & $v_6$: $[1,5], [6,\infty]$ \\
    $v_2$: $[1,3], [3,5], [4,\infty]$ & $v_7$: $[1,5], [6,\infty]$ \\
    $v_3$: $[1,4], [2,6], [3,7], [4,\infty]$ & $v_8$: $[1,5], [4,\infty]$ \\
    $v_4$: $[1,3], [3,5], [4,\infty]$ & $v_9$: $[1,4], [2,\infty]$ \\
    $v_5$: $[1,7], [7,\infty]$ & \\
  \bottomrule
  \vspace{1pt}
\end{tabular}
\end{table}

% \begin{table}[tbp]
%     \caption{The vertex core time index of the temporal graph $G$ \reffig{example} for $k=2$.}
%     \label{tab:vct}
%     \centering
%     % \begin{tabular}{ | >{\raggedright\arraybackslash}m{3.8cm} | >{\raggedright\arraybackslash}m{3.8cm} | }
%     \begin{tabular}{ l l }
%         \hline
%         $v_1$: $[1,3], [3,5], [6,7], [7,\infty]$ & $v_6$: $[1,5], [6,\infty]$ \\
%         \hline
%         $v_2$: $[1,3], [3,5], [4,\infty]$ & $v_7$: $[1,5], [6,\infty]$ \\
%         \hline
%         $v_3$: $[1,4], [2,6], [3,7], [4,\infty]$ & $v_8$: $[1,5], [4,\infty]$ \\
%         \hline
%         $v_4$: $[1,3], [3,5], [4,\infty]$ & $v_9$: $[1,4], [2,\infty]$ \\
%         \hline
%         $v_5$: $[1,7], [7,\infty]$ & \\
%         \hline
%     \end{tabular}
% \end{table}

\begin{example}
\reftab{vct} presents the vertex core time index of the temporal graph $G$ for $k=2$ in \reffig{example}. For the vertex $v_1$, $[1,3]$ indicates that the core time of $v_1$ is $3$, and the vertex is in the $2$-core of all windows staring from $1$ ending by a time not earlier than $3$. The next label of $[1,3]$ is $[3,5]$. This implies the core time of $v_1$ for $ts = 2$ is still $3$, and its core time for $ts = 3$ is 5.
\end{example}

The vertex core time index can be computed in $O(|\vctidx|*\avgdegree)$ time complexity \cite{Yu2021}, where $|\vctidx|$ is the size of the vertex core time index and $\avgdegree$ is the average vertex degree in the graph. We will discuss the basic idea of computing $\vctidx$ in \refsubsec{vertexcoretime}.
The vertex core time index essentially compresses vertices in $k$-cores of all possible windows to some extent. This motivates us to develop algorithms to transfer it into all temporal $k$-cores. 
Given that each temporal $k$-core is distinguished by their edges \cite{Yang2023}, our idea is to first compute a structure to compress all possible windows for each edge $e$ such that $e$ is in the $k$-core of the window. Then, we utilize the windows of each edge and assemble them to produce the final temporal $k$-cores.
To this end, we first derive the \textit{edge core window skyline} of each edge, denoted as $\edgeskyline$, as a by-product of computing vertex core times without increasing the time complexity. The core window skyline of each edge is the set of all minimal core windows defined as follows.

% the core window skyline of all edges as a byproduct of computing vertex core times, which never increases the existing time complexity of computing $\vctidx$. The core window skyline of each edge is the set of all minimal core windows defined as follows.

% We then transfer all edge core times to a equivalent structure, called minimal core windows for all edges, which is defined as follows.

\begin{definition}[Minimal Core Window]
Given a temporal graph $G$, an integer $k$, and an edge $e$, a time window $[ts,te]$ is a minimal core window of $e$ if (1) $e$ is in the $k$-core of $G_{[ts,te]}$; and (2) $e$ is not in the $k$-core of $G_{[ts',te']}$ for any $[ts',te'] \subset [ts,te]$.
\end{definition}

\begin{table}[tbp]% h asks to places the floating element [h]ere.
    \centering
    \caption{The minimal core windows (edge core window skyline) of all edges in the temporal graph $G$ of \reffig{example} for $k = 2$}
    \label{tab:ecs}
  \begin{tabular}{ll}
    \toprule
    $(v_2,v_9,1)$: $[ 1,4 ]$ & $(v_1,v_6,5)$: $[ 5,5 ]$ \\
    $(v_1,v_4,2)$: $[ 2,3 ]$ & $(v_1,v_7,5)$: $[ 5,5 ]$ \\
    $(v_2,v_3,2)$: $[ 1,4 ], [ 2,6 ]$ & $(v_2,v_8,5)$: $[ 3,5 ]$ \\
    $(v_1,v_2,3)$: $[ 2,3 ], [ 3,5 ]$ & $(v_6,v_7,5)$: $[ 5,5 ]$ \\
    $(v_2,v_4,3)$: $[ 2,3 ], [ 3,5 ]$ & $(v_1,v_3,6)$: $[ 2,6 ], [ 6,7 ]$\\
    $(v_3,v_9,4)$: $[ 1,4 ]$ & $(v_3,v_5,6)$: $[ 6,7 ]$\\
    $(v_4,v_8,4)$: $[ 3,5 ]$ & $(v_1,v_5,7)$: $[ 6,7 ]$\\
  \bottomrule
\end{tabular}
\end{table}

% \begin{table}[tbp]
%     \centering
%     \begin{tabular}{ l l l }
%         \hline
%         $(v_2,v_9,1)$: $[ 1,4 ]$ & $(v_1,v_6,5)$: $[ 5,5 ]$ \\
%         \hline
%         $(v_1,v_4,2)$: $[ 2,3 ]$ & $(v_1,v_7,5)$: $[ 5,5 ]$ \\
%         \hline
%         $(v_2,v_3,2)$: $[ 1,4 ], [ 2,6 ]$ & $(v_2,v_8,5)$: $[ 3,5 ]$ \\
%         \hline
%         $(v_1,v_2,3)$: $[ 2,3 ], [ 3,5 ]$ & $(v_6,v_7,5)$: $[ 5,5 ]$ \\
%         \hline
%         $(v_2,v_4,3)$: $[ 2,3 ], [ 3,5 ]$ & $(v_1,v_3,6)$: $[ 2,6 ], [ 6,7 ]$\\
%         \hline
%         $(v_3,v_9,4)$: $[ 1,4 ]$ & $(v_3,v_5,6)$: $[ 6,7 ]$\\
%         \hline
%         $(v_4,v_8,4)$: $[ 3,5 ]$ & $(v_1,v_5,7)$: $[ 6,7 ]$\\
%         \hline
%     \end{tabular}
%     \caption{The minimal core windows (edge core window skyline) of all edges in the temporal graph $G$ of \reffig{example} for $k = 2$}
%     \label{tab:ecs}
% \end{table}

\begin{figure}[t!]
\centering
\includegraphics[scale=0.35]{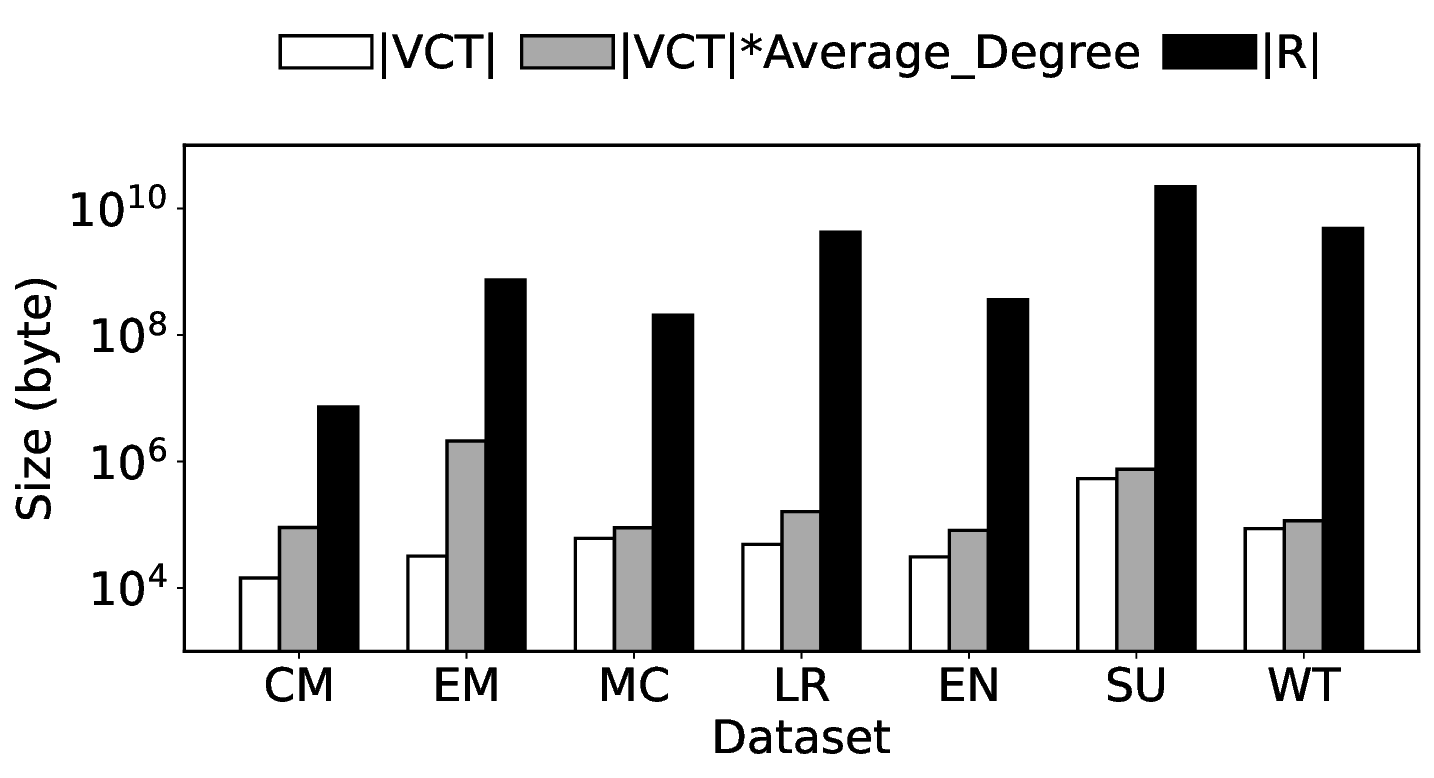}
\caption{$|\vctidx|$, $|\vctidx| * \avgdegree$, and $|\rset|$ for seven representative datasets given default parameters ($k$ = 30\% $k_{max}$, $t$ = 10\% $t_{max}$).}
\label{fig:ct_vs_ressize}
\end{figure}

\begin{example}
\reftab{ecs} presents the $\edgeskyline$ for all edges in the graph of \reffig{example} for $k = 2$. The edge $(v_2,v_9)$ has one minimal core window $[1,4]$, which indicates that $(v_2,v_9)$ is contained in a $2$-core in the time window $[1,4]$. For any sub-window of $[1,4]$, $(v_2,v_9)$ is not contained in any $2$-cores in that sub-window.
\end{example}

Based on the edge core window skyline of all edges, we develop an algorithm for temporal $k$-core enumeration and bound the time complexity by the result size, which is optimal in this phase. The result size means the sum of the number of edges in all resulting temporal $k$-cores. In this way, we will achieve an overall time complexity of $O(|\vctidx|*\avgdegree+|\rset|)$, where the first part $|\vctidx|*\avgdegree$ is the time complexity to compute the edge core window skyline and $|\rset|$ is the result size.

\begin{remark}
\reffig{ct_vs_ressize} presents $|\vctidx|$, $|\vctidx| * \avgdegree$, and $|\rset|$ for several representative datasets given the default query parameters in our experiments. All other datasets evaluated in our experiments present the similar trend. The sizes of the result sets are 2 to 4 orders of magnitude larger than $|\vctidx|*\avgdegree$ for all datasets. This indicates that the time complexity of our final algorithm is mainly related to the result size $|\rset|$ in practice.
\end{remark}

% Figure \ref{fig:overview} presents an overview of the framework. Given an input temporal graph and an integer $k$, we first compute the Core Time Index for all vertices. The time complexity is $O(|\ctidx| * \mathcal{D}_{avg})$, where $\mathcal{D}_{avg}$ is the average degree. We then convert it into the Core Time Skyline Index, which takes $O(|\ctidx|)$. Finally, we derive all distinct Temporal K-vertex Sets $\vset$ from the Core Time Skyline Index. The time complexity is $|\rset|$, which is the size of the result set.

% We can alternatively compute the Core Time Index for all edges to derive all distinct Temporal $k$-cores $\eset$. The same procedure is then applied to the Core Time Index of edges. 

\subsection{Computing Edge Core Window Skyline}
\label{subsec:vertexcoretime}

\begin{algorithm}[t]
\caption{\algct}
\label{alg:ct_comp}
\KwIn{$G(V,E), Ts, Te, k$}
\KwOut{$\edgeskyline$}

compute the vertex core time index by \cite{Yu2021}\;

% $\ct(u) \gets$ compute core time of each vertex $u$ for the start time $Ts$ and $k$ by iteratively removing edges from $Te$ to $Ts$\;

\tcc{initialize edge core times}
\ForEach{edge $(u,v,t)$ in $G_{[Ts,Te]}$}{
    $\ct(u,v) \gets \max(\ct_{Ts}(u),\ct_{Ts}(v),t)$\;
    $\edgeskyline(u,v) \gets \emptyset$\;
}
% \ForEach{$u \in V$}{
%     $CTD(u) \gets |\{\pair{v,t} \in \nbr(u)_{[Ts,Te]}| \max(\ct(v),t) \le \ct(u)\}|$\;
% }

\ForEach{$T_s \leq t < T_e$}{
    \ForEach{vertex $u$ w.r.t. $\ct_{t+1}(u) \neq \ct_{t}(u)$}{
        \ForEach{$\pair{v,t'} \in N(u)_{[t+1, T_e]}$} {
            \tcc{update edge core times}
            $new\_ct \gets \max(\ct_{t+1}(u),\ct_{t+1}(v),t')$\;
            \If{$new\_ct > \ct(u,v)$}{
                $\edgeskyline(u,v) \gets \edgeskyline(u,v) \cup \{[t,\ct(u,v)]\}$\;
                $\ct(u,v) \gets new\_ct$\;
            }
            
        }
    }
    % $Q \leftarrow \emptyset$\;
    % \ForEach{$(u,v) \in E_t$} {
    %     \lIf{$\ct(u) = \infty \lor \ct(v) = \infty$}{\kwcontinue}
    %     \lIf{$\ct(v) \le \ct(u)$}{$CTD(u) \gets CTD(u) - 1$}
    %     \lIf{$CTD(u) < k$}{$Q \leftarrow Q \cup \{u\}$}
    %     perform lines 11--12 by swapping $u$ and $v$\;
    % }
    % \While{$Q \neq \emptyset$} {
    %     $u \leftarrow Q.pop()$\;
    %     $old\_ct \gets \ct(u)$\;
    %     $\ct(u) \gets$ compute the core time of $u$ for the start time $t+1$ based on \reflem{ct}\;

    %     recompute $CTD(u)$ in $\nbr(u)_{[t+1,Te]}$\;

    %     \ForEach{$\pair{v,t'} \in N(u)_{[t+1, T_e]}$} {
    %         \If {$max(old\_ct, t') \leq \ct(v) < \ct(u)$}{
    %             $CTD(v) \leftarrow CTD(v) - 1$\;
    %             \lIf {$CTD(v) < k$} {$Q \leftarrow Q \cup \{v\}$}
    %         }

    %         \tcc{update edge core times}
    %         \If{$\max(\ct(u),\ct(v),t') > \ct(u,v)$}{
    %             $\edgeskyline(u,v) \gets \edgeskyline(u,v) \cup \{[t,\ct(u,v)]\}$\;
    %             $\ct(u,v) \gets \max(\ct(u),\ct(v),t')$\;
    %         }
            
    %     }

    % }
}
\Return{$\edgeskyline$}
\end{algorithm}

% We review the existing technique \cite{Yu2021} to compute the vertex core time index in this section and extend it to derive the edge core window skyline.

We first review the existing technique \cite{Yu2021} to compute the vertex core time index. Given the query window $[Ts,Te]$, they first compute the core time of each vertex for the start time $Ts$. Then, they update the core time by increasing the start time from $Ts$ to $Te$ and record all distinct core time values. To this end, certain values for each vertex are maintained when increasing the start time. Monitoring these values enables identifying whether the core time of any vertex changes. If so, the core time is recomputed for the current start time. They scan neighbors of a node to recompute the vertex core time for any start time, and no extra cost is required to monitor if the core time of any vertex needs change. As a result, they achieve $O(|\vctidx|*\avgdegree)$ time complexity.

\stitle{Deriving Edge Core Window Skyline.} To compute the edge core window skyline, we first demonstrate a correspondence between the vertex core time and the edge core time. In this way, we update the edge core time when any vertex core time changes. 

\begin{lemma}
\label{lem:vtoe}
Given a start time $ts$ and an integer $k$, the core time of an edge $(u,v,t)$ is $\ct_{ts}(u,v,t) = \max(\ct_{ts}(u),\ct_{ts}(v),t)$.
\end{lemma}

Based on \reflem{vtoe}, we initialize the edge core time for the start time $Ts$ in lines 3. $\edgeskyline(u,v)$ in line 4 is used to collect all minimal core windows for $(u,v)$.
Then we update the core time for each edge in line 7 if the core time of any terminal of the edge updates. We present the following lemma to derive the edge minimal core window when the edge core time updates.

\begin{lemma}
\label{lem:ct2minimal}
Given a start time $ts$ and an integer $k$, assume the core time of an edge $e = (u,v,t)$ for $ts$ is different from that for $ts+1$. We have $[ts,\ct_{ts}(e)]$ is a minimal core window of $e$ for $k$.
\end{lemma}

\begin{proof}
We prove this lemma by contradiction. Suppose we have $\ct_{ts+1}(e) = \ct_{ts}(e)$ and $[ts,\ct_{ts}(e)]$ is not a minimal core window of $e$, then $e$ is within the k-core of the time window $[ts+1, \ct_{ts+1}(e)] \subset [t_s, \ct_{ts+1}(e)]$, which contradicts the definition of minimal core window.
\end{proof}

Based on \reflem{ct2minimal}, we derive a minimal core window of the edge in line 10 when its core time updates.

\begin{example}
Refer to the graph in \reffig{example}. We compute the $\edgeskyline$ for $k = 2$ by \refalg{ct_comp}. Initially at $ts = 1$, $\ct_1(v_2) = 3$ and $\ct_1(v_3) = 4$. Therefore by line 3, we have $\ct_1(v_2,v_3) = \max(3, 4, 2) = 4$. As we move onto $ts = 2$, the core time of $v_3$ has changed to $\ct_2(v_3) = 6$, which is larger than the existing core time of the edge $\ct_1(v_2,v_3) = 4$. Therefore, we have $\ct_2(v_2, v_3) = 6$. Since the core time of $(v_2, v_3)$ has changed at $ts = 2$, we add $[1, 4]$ into $\edgeskyline(v_2, v_3)$ by line 10. Then, we update the core time of $(v_2, v_3)$ to 6 by line 11.
\end{example}

Processing the edge core time (lines 2--4 and lines 9--11) does not incur additional time complexity. We achieve the same time complexity of computing vertex core time in \cite{Yu2021}, which is formally summarized as follows.

\begin{theorem}
\label{thm:cttime}
The time complexity of \refalg{ct_comp} is $O(|\vctidx| \cdot \avgdegree)$.
\end{theorem}

%% file: 5enum.tex
\section{Enumerating Temporal K-Cores}
\label{sec:solution}

\subsection{A Straightforward Method}
\label{subsec:naive}

We first discuss a basic solution to enumerate all temporal $k$-cores given the minimal core windows of all edges. Similar to OTCD \cite{Yang2023}, our algorithm enumerates all time windows within the query time range. Given each time window, the following lemma demonstrates a way to derive its temporal $k$-core based on minimal core windows of edges.

\begin{lemma}
\label{lem:allwindow}
Given a time window $[ts,te]$ and its temporal $k$-core $C$, an edge $e$ is in $C$ if and only if there is a minimal core window $[t1,t2]$ of $e$ contained in $[ts,te]$, i.e., $[t1,t2] \subseteq [ts,te]$.
\end{lemma}

Based on \reflem{allwindow}, the union of all edges satisfying the condition is the temporal $k$-core of $[ts,te]$. 
Algorithm \ref{alg:baseline} presents the basic solution. For each start time $ts$ from $Te$ to $Ts$, we initialize a set of buckets (line 3). Then, we obtain the relevant window of all edges that satisfies the condition in \reflem{allwindow} and insert the corresponding edge into the designated bucket based on the end time of the window (lines 4-6). After the bucket is constructed, for each end time $te$, we insert all vertices in the bucket to the current edge set (line 10). A minor optimization is applied in line 9, where $B[te] = \emptyset$ means the temporal $k$-core is the same as that for $[ts,te-1]$.
Then, we check if the current temporal $k$-core is already in the result set in line 11. To this end, we maintain all computed temporal $k$-cores by a hash table in $\rset$. We add the temporal $k$-core into the result set in line 12. 
By utilizing the edge core window skyline, we are able to find all temporal $k$-cores for each time window.

\stitle{Drawbacks of \refalg{baseline}.} The improvement opportunities of \refalg{baseline} lie in two aspects. First, even with an optimization (line 9) applied to reduce the examination of unnecessary time windows, the algorithm still needs to scan $O(\tmax^2)$ time windows in the worst case. Second, the same temporal $k$-core may be computed when processing multiple time windows. This implies significant costs for unnecessary computations.

% \rr{A major drawback of this solution is that it requires all previously derived results to be stored, which is impractical for large graphs.} Moreover, this algorithm requires to examine each time window, which incurs a time complexity of $O(|t|^2)$, where \rr{$|t|$ is the number of distinct time stamps}. For each time window, it is required to check for the validity of each induced sub-graph and hash it by sorting the edges, which are both time-consuming processes. The high cost of such a checking process for each time window makes the algorithm impractical in large temporal graphs with a large number of timestamps.

\begin{algorithm}[t!]
\caption{\algbaseenum}
\label{alg:baseline}
\SetAlgoVlined
\KwIn{a temporal graph $G$, an integer $k$, a time range $[Ts,Te]$, the $\edgeskyline$ of $G$ for $k$}
\KwOut{all distinct temporal $k$-cores in $[Ts,Te]$}
$\rset \gets \emptyset$\;

\ForEach{$Ts \le ts \le Te$} {
    \lForEach {$ts \le te \le Te$}{$B[te] \gets \emptyset$}
    \ForEach {edge $e \in \edgeskyline$} {
        $[t1,t2] \gets$ the first window in $\edgeskyline(e)$ such that $t1 \ge ts$\;

        $B[t2] \gets B[t2] \cup \{(e)\}$

    }
    $C \gets \emptyset$\;

    \ForEach{$ts \le te \le Te$} {
        \lIf {$B[te] = \emptyset$}{\kwcontinue}
        $C \gets C \cup B[t_e]$\;

        \lIf{$C \in \rset$}{\kwcontinue}
        $\rset \gets \rset \cup \{C\}$\;
    }
}
\Return{$\rset$}
\end{algorithm}

\subsection{Anchoring the Start Time}
\label{subsec:start}

% We start by showing several properties for the correspondence between the edge core skyline and temporal $k$-cores. 
% \stitle{Notations.}

To respond to the drawbacks of the basic solution, we expect to design an algorithm that (1) only scans time windows with an undiscovered temporal $k$-core and (2) visits each temporal $k$-core only once.
Given the edge core window skyline in this section, we first discuss the data structure and the corresponding algorithm to enumerate all temporal $k$-cores for a certain start time, which is a subproblem of temporal $k$-core enumeration. We will discuss how to update the structure from one start time to the next start time in \refsubsec{all}, which produces all final results.

Based on \refdef{tti}, there is a one-to-one correspondence between a temporal $k$-core and its Tightest Time Interval (TTI). In other words, if we find TTIs of all temporal $k$-cores, we can output the temporal $k$-core of each TTI as a result. Recall that the temporal $k$-core of a time window (TTI) $[ts,te]$ can be computed by the union of all edges with a minimal core window in $[ts,te]$ (\reflem{allwindow}). Therefore, we mainly discuss the relation between TTIs and minimal core windows below, which produces a method to derive TTIs based on minimal core windows.
For simplicity, we say the start time (resp. the end time) of a temporal $k$-core $C$ to represent the start time (resp. the end time) of the TTI of $C$. Note that the temporal $k$-core $C$ of a time window $[ts,te]$ does not guarantee $[ts,te]$ is a TTI of $C$ (the TTI may be a subwindow of $[ts,te]$).
We start by showing that not all possible time labels can be the start time of a temporal $k$-core.

% TODO

% not all start time need to consider

\begin{lemma}
\label{lem:windowstart}
A temporal $k$-core starting from $ts$ exists if and only if there exists a minimal core window $[t1,t2]$ of an edge with $t1 = ts$.
\end{lemma}

\begin{proof}
We first assume a temporal $k$-core exists with a TTI $[ts,te]$. There must exist a minimal core window starting from $ts$. Otherwise, we can increase the start time of the temporal $k$-core to the earliest value $ts'$ such that a minimal core window starts from. The edges in the temporal $k$-core of $[ts',te]$ is the same as that of $[ts,te]$ based on \reflem{allwindow}, which contradicts the TTI $[ts,te]$.
Next, given a minimal core window $[t1,t2]$ of an edge $e$, let $C$ be the temporal of $k$-core of the window $[t1,t2]$. It is clear to see that $[t1,t2]$ is the TTI of $C$ because $e$ is not in the temporal $k$-core of any subwindow of $[t1,t2]$.
\end{proof}

% \begin{lemma}
% \label{lem:windowstart1}
% Given a temporal $k$-core $C$ and the TTI $[ts,te]$ of $C$, there exists a minimal core window $[t1,t2]$ of an edge such that $t1 = ts$ and $t2 \le te$.
% \end{lemma}

% \begin{lemma}
% \label{lem:windowstart2}
% Given an arbitrary minimal core window $[t1,t2]$ for an edge, there exists a temporal $k$-core whose TTI starts from $t1$. 
% \end{lemma}

\reflem{windowstart} provides a necessary and sufficient condition for the start time of any temporal $k$-core. For simplicity, we call a time $ts$ a \textit{valid start time} if there exists a minimal core window $[t1,t2]$ for an edge with $t1=ts$. 
% Now, given a valid start time $ts$, assume we already know that a temporal $k$-core exists with a TTI $[ts,te]$. \reflem{allwindow} indicates that all edges with a minimal core window in $[ts,te]$ must belong to the temporal $k$-core. 
Next, given the valid start time $ts$, we first discuss several properties of the end time $te$ such that a temporal $k$-core exists with the TTI $[ts,te]$. Then, we discuss the structure to efficiently enumerate edges of temporal $k$-core for different valid end times.

\stitle{Valid End Times.} We propose two necessary conditions for the end time of temporal $k$-cores as follows.

% \begin{lemma}
% \label{lem:validend}
% Given a start time $ts$ such that a minimal core window starting from $ts$ exists, a temporal $k$-core exists with the TTI $[ts,te]$ if and only if there exists minimal core window $[t1,t2]$ for an edge $e$ such that (1) $t2=te$; and (2) there does not exists a minimal core window $[t1',t2']$ for $e$ with $ts \le t1' < t1$.
% \end{lemma}

\begin{lemma}
\label{lem:validend1}
A temporal $k$-core exists with a TTI $[ts,te]$ only if there exists a minimal core window $[t1,t2]$ for an edge $e$ such that (1) $ts \le t1, t2=te$; and (2) there does not exist a minimal core window $[t1',t2']$ for $e$ with $ts \le t1' < t1$.
\end{lemma}

\begin{proof}
We prove both conditions by contradiction. 
For the first condition, suppose we have a temporal $k$-core with a TTI $[ts, te]$ and there does not exist a minimal core window $[t1, t2]$ such that $ts \leq t1$ and $t2 = te$. Based on \reflem{allwindow}, there must exist a minimal core window of an edge in the TTI. Therefore, we can decrease $te$ until we find a window $[t1, t2]$ satisfying $ts \leq t1$ and $t2 = te$. As a result, we derive the same temporal $k$-core with a tighter TTI, which contradicts the initial assumption.
%
% If all windows have $t1 < ts$, there will be no edges in the $k$-core of $[ts, te]$, a contradiction. If there is no window with $t2 = te$, then based on the definition of the minimal core window, for the temporal $k$-core to exist in $[ts, te]$, there must exist a minimal core window with $ts \leq t1$ and $t2 < te$. However, this indicates that an identical temporal $k$-core exists with TTI $[ts, t2]$. Since $[ts, t2] \subset [ts, te]$, $[ts, te]$ cannot be a TTI of the same temporal $k$-core, a contradiction. 
%%
For the second condition, suppose we have a set of minimal core windows satisfying condition 1. For each of such windows, there exists another window $w' = [t1', t2']$ such that $ts \leq t1' < t1$. Based on the definition of $\edgeskyline$, we have $t2' < t2$. Therefore, the same temporal $k$-core exists by decreasing the end time of TTI based on \reflem{allwindow}, which produces a contradiction.
\end{proof}

The minimal core window $[t1,t2]$ in \reflem{validend1} is clearly the earliest one among all minimal core windows of $e$ starting not earlier than $ts$. \reflem{validend1} implies that only one minimal core window for each edge is required to enumerate all temporal $k$-cores starting from $ts$. Motivated by this, we define the active time of each minimal core window as follows to indicate whether the window should be considered for a specific start time.

\begin{definition}[active time]
\label{def:act}
Given a set of minimal core windows of an edge $e$, let the windows be ordered by increasing $t1$ (and consequently increasing $t2$). For a specific minimal core window $w = [t1,t2]$, let $w' = [t1',t2']$ the immediately preceding window in this order, such that $t1' < t1$. If such a preceding window $w'$ exists, the activation time of $w$, denoted as $w.active$, is $t1' + 1$. Otherwise, $w.active = 1$.
\end{definition}

Based on \refdef{act}, the condition (2) in \reflem{validend1} can be replaced by $[t1,t2].\yestime \le ts$.

\begin{example}
Consider the graph $G$ in \reffig{example} and the corresponding edge core window skyline in \reftab{ecs}. For the minimal core window $[3,5]$ of the edge $(v_1,v_2,3)$, its active time is $3$. That means we do not need to consider the window $[3,5]$ for the temporal $k$-core with the TTI starting from $1$ or $2$.
\end{example}

To derive the temporal $k$-core starting from $ts$, \reflem{validend1} and \refdef{act} motivate us to collect all minimal edge windows with active times not later than $ts$ and start times not earlier than $ts$. Then, we only consider the end times of the windows as the potential end times of temporal $k$-cores. 
Next, we present the second property to further filter end times by minimal core time windows of edges.

% \begin{lemma}
% \label{lem:et=ts}
% \rr{For a start time $ts$, let $t'$ be the earliest end time such that $[ts, t2]$ is a minimal core window in $\edgeskyline$, there exists at least one window such that its corresponding edge has a timestamp of $ts$, and $[ts, t']$ is the last window of all such edges.}
% \end{lemma}

% \begin{proof}
% \rr{We prove this lemma by contradiction. Suppose we have a set $S$ of windows such that for each $w \in S$, we have $w = [ts,t']$ and $w.edge.t > ts$. Since $t'$ is the earliest time such that $[ts, t']$ is a window within $\edgeskyline$, the set of corresponding edges $\{w.edge | w \in S\}$ are the only edges that contributes to the temporal $k$-core at the end time $t'$. Say the largest timestamp of these edges is $t$ where $t > ts$, we have an identical temporal $k$-core within the time-span of $[t, t']$. Therefore, $[ts,t']$ is not a valid window in $\edgeskyline$, a contradiction.}
% \end{proof}

% \rr{Based on \reflem{et=ts}, we present the second condition to further filter end times.}

% TODO

\begin{lemma}
\label{lem:validend2}
A temporal $k$-core exists with a TTI $[ts,te]$ only if there exists a minimal core window $[t1,t2]$ of an edge $e$ such that $t1=ts, t2 \le te$.

%(1) $t1=ts, t2 \le te$; (2) there does not exist a minimal core window $[t1',t2']$ of $e$ with $t2 < t2' \le te$.
% there exists a minimal core window $[t1,t2]$ such that $t1=ts$, $t2 \le te$, and $[t1,t2].\notime > te$.
\end{lemma}

% \begin{proof}
% The proof is similar to that of \reflem{windowstart}, and the details are omitted here.
% \end{proof}

Recall that there must exist at least one minimal core window starting from the same time of a temporal $k$-core based on \reflem{windowstart}. \reflem{validend2} indicates that the end time of the temporal $k$-core is not earlier than the earliest end time of those windows.
Below, we show that the aforementioned two necessary conditions are sufficient for the end time of temporal $k$-cores.

\begin{theorem}
\label{thm:validend3}
A temporal $k$-core with a TTI $[ts,te]$ exists if conditions in both \reflem{validend1} and \reflem{validend2} hold.
\end{theorem}

\stitle{Proof of \refthm{validend3}.}
Based on either \reflem{validend1} or \reflem{validend2}, there must exist a minimal core window contained in $[ts,te]$. This proves that a temporal $k$-core $C$ exists in $[ts,te]$. Next, we prove that $[ts,te]$ is the TTI of the temporal $k$-core $C$. To this end, let $C'$ be the temporal $k$-core of $[ts,te]$ satisfying \reflem{validend1} and \reflem{validend2}. We show that narrowing the window will remove edges from $C'$.
First, let $e$ be the edge meeting the condition in \reflem{validend1}. It is clear to see that $e$ is in $C'$, and decreasing the end time $te$ of the window will exclude $e$ from $C'$.
Next, we show that increasing the start time $ts$ of the window will also exclude an edge from $C'$ by the following lemma.

% \reflem{validend1} ensures that a valid start time $ts$ exists within the query time range, while \reflem{validend2} restricts the end time $te$ to a feasible range, ensuring the time window is well-defined. Together, these conditions ensure that $[ts,te]$ satisfies the required properties for the theorem.

\begin{lemma}
\label{lem:onestart}
Let $[t1,t2]$ be a minimal core window such that there does not exist a minimal core window $[t1',t2']$ with $t1=t1'$ and $t2 > t2'$. An edge $e=(u,v,t1)$ exists and $[t1,t2]$ is a minimal core window of $e$.
\end{lemma}

\begin{proof}
Assume the edge $e$ does not exist. Removing all edges at $ts$ would not change the temporal $k$-core of $[ts,te]$. All edges with the minimal core window $[t1,t2]$ are still in the temporal $k$-core of $[t1+1,t2]$, which contradicts that $[t1,t2]$ is a minimal core window.
\end{proof}

Given of all minimal core windows satisfying the condition in \reflem{validend2}, let $[t1',t2']$ be one of them (i.e., $t1'=ts, t2' \le te$) with the earliest end time. \reflem{onestart} indicates there must exist an edge $e$ at $ts$ with a minimal core window $[t1',t2']$. Therefore, $e$ is in $C'$, but increasing $ts$ will exclude $e$ from the projected graph and the temporal $k$-core $C'$. We finish the proof of \refthm{validend3}.

\begin{figure}[tbp]
\centering
    \subfigure[$ts = 1$.]{
        \label{fig:L1}
        \includegraphics[width=85mm]{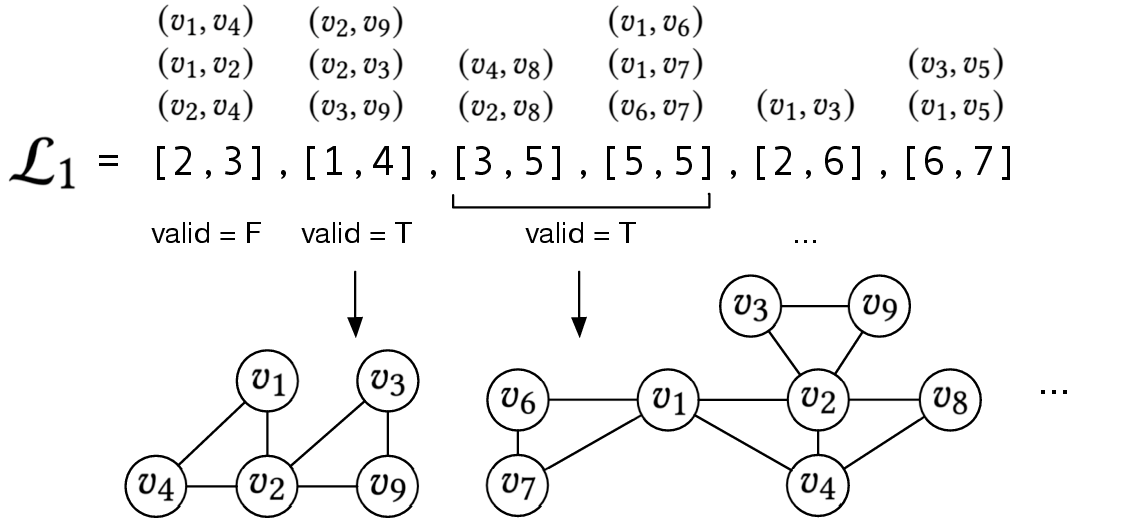}
    }
    \subfigure[$ts = 2$.]{
        \label{fig:L2}
        \includegraphics[width=80mm]{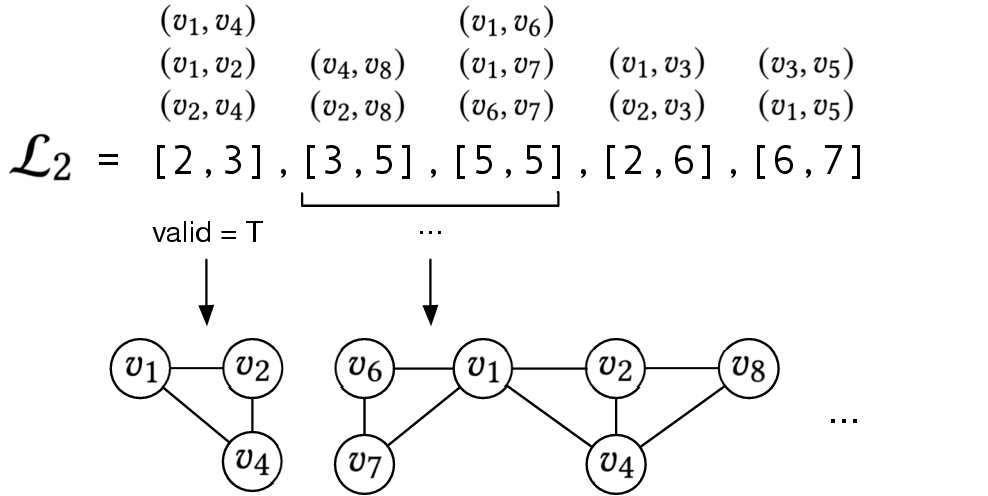}
    }
\caption{Illustration of enumerating $\dll_{ts}$ for $ts = 1$ and $ts = 2$.}
\label{fig:DLL}
\end{figure}

\stitle{The Structure \& Algorithm.} We now have necessary and sufficient conditions for the end time given a valid start time $ts$, which enables us to present the enumeration algorithm. To check the condition in \reflem{validend2} efficiently given the start time $ts$, we maintain the following data structure for enumerating all temporal $k$-cores starting from $ts$. The data structure, denoted by $\dll_{ts}$, includes all minimal core windows with active times (\refdef{act}) not later than $ts$. To improve the enumeration efficiency, we arrange all windows in $\dll_{ts}$ by ascending order of end times. 
% \rr{By doing so, once we identified a temporal $k$-core with start time $ts$ due to a window satisfying the conditions in \reflem{validend2}, the identification of all subsequent temporal $k$-cores with the same start time can be simplified.}

\begin{example}
    Given the temporal graph $G$ in \reffig{example} and the minimal core windows of all edges for $k=  2$ in \reftab{ecs}, \reffig{L1} shows the data structure $\dll_{ts}$ for $ts = 1$ used in our algorithm. All minimal core windows with activate times not later than $1$ are included in the structure and arranged in a ascending order of end times. Specifically, the minimal core windows $[2,3]$ of $(v_1,v_4)$, $(v_1,v_2)$ and $(v_2,v_4)$ are ranked as the first three windows in $\dll_1$ since they have the earliest end time of 3 among all windows in $\dll_{1}$. The minimal core window $[1,4]$ of edges $(v_2, v_9)$, $(v_2, v_3)$ and $(v_3, v_9)$ are ranked as the fourth to sixth windows in $\dll_1$, as their end times are the second smallest. Note that each edge has at most one minimal core window in $\dll_{1}$.
\end{example}

% \begin{definition}[EA Order]
% Given a set of minimal core windows, an EA order is 

% that for any two windows $[t1,t2]$ and $[t1',t2']$ in the order, $[t1,t2]$ ranks higher $[t1',t2']$ if (1) $t2 < t2'$; or (2) $t2=t2'$ and $t1,t2$ is auxiliary.
% \end{definition}

\begin{algorithm}[t]
\caption{$\algstartanchor$}
\label{alg:start}
\KwIn{$ts$, $\dll_{ts}$}
\KwOut{distinct temporal cores for a specific start time $ts$}
$w \gets$ the first window in $\dll_{ts}$\;
$R \leftarrow \emptyset$\;
$valid \gets \kwfalse$\;
\While{$w \neq \kwnull$}{
    $R \gets R \cup \{w.edge\}$\;
    \lIf{$w.start = ts$}{$valid \gets \kwtrue$}
    \If{$valid = \kwfalse$}{
        $w \gets w.next$\;
        \textbf{continue}\;
    }
    \If{$w.next \neq \kwnull \land w.end = w.next.end$}{
        $w \gets w.next$\;
        \textbf{continue}\;
    }

    \textbf{output} $R$\;
    $w \gets w.next$\;
}
\end{algorithm}

Given the structure $\dll_{ts}$, we present the algorithm to enumerate all temporal $k$-cores starting from $ts$ in \refalg{start}. We will introduce how $\dll_{ts}$ is precomputed and maintained for different start times in \refsubsec{all}. Given a minimal core window $w$, $w.start$ and $w.end$ denote the start time and the end time of $w$, respectively. $w.next$ denotes the next window of $w$ in $\dll_{ts}$, and $w.edge$ denotes the corresponding edge of the window. We iteratively scan windows in $\dll_{ts}$ and add the corresponding edge into the edge set (line 5).
We use a flag $valid$ to check the condition in \reflem{validend2}. Suppose window $w^i$ is the first window we visit with star time $t1 = ts$. We can update the flag to $\kwtrue$ as the condition \reflem{validend2} holds. If the flag is $\kwfalse$ in line 7, we do not output the current result. The condition in line 10 means we have more windows with the same end time. We continue to the next one (line 11) and only output the result for the last window with the same end time (line 13).

\begin{example}
For the temporal graph in \reffig{example}, we enumerate all temporal 2-cores with $ts = 1$ by scanning $\dll_1$ (\reffig{L1}). The first three minimal core windows $[2,3]$ have start times not equal to $ts$, so their edges are added to $R$ without output, and the flag $valid$ remains False. The next three windows $[1,4]$ have start times of 1; their edges are added to the edge-set, which is output as a result with $TTI = [1,4]$, setting $valid$ to $\kwtrue$ for the remainder of the iteration. Each subsequent matching window updates the edge-set, outputting distinct results. For $ts = 2$ (\reffig{L2}), we scan $\dll_2$. The first three windows meet the condition, and their edge-set is output, representing the subgraph skipped in the previous round ($ts = 1$). The scan continues until the end of $\dll_2$. 

% For the temporal graph in \reffig{example}, we enumerate all temporal $2$-cores with start time $ts = 1$. To this end, we scan through $\dll_1$. We start by visiting the first three minimal core windows $w = [2,3]$, and none of them has the start time $w.start$ of $1$. Therefore, we add their corresponding edges to the $R$ but do not output them. The flag $valid$ remains as $\kwfalse$. We then visit the next three windows of $w = [1,4]$. Since their start times are 1, which equals $ts$, we set $valid = \kwtrue$, add the three corresponding edges to the edge-set, and output the edge-set consisting of all six edges as a result with $TTI = [1,4]$. From now on, the flag is set to $\kwtrue$ for the rest of the iteration. When we visit certain minimal core windows in each step, we add them into the edge-set and output the edge-set as a distinct result.
% %
% For $ts = 2$ in \reffig{L2}, we scan through $\dll_2$. For the first three windows, the condition is met, and we output the edge-set. Note that the sub-graph outputted is the one we skipped in the previous round when $ts = 1$. The scan continues until we reach the end of $\dll_2$.
\end{example}

\begin{lemma}
Let $|\rset_{ts}|$ be the total size of all temporal $k$-cores starting from $ts$. Given the data structure $\dll_{ts}$, the time complexity of \refalg{start} is $O(|\rset_{ts}|)$.
\end{lemma}

\begin{algorithm}[tb!]
\caption{\algenum}
\label{alg:enum}
\KwIn{a temporal graph $G$, an integer $k$, a time range $[Ts,Te]$, the $\edgeskyline$ of $G$ for $k$}
\KwOut{all distinct temporal $k$-cores in $[Ts,Te]$}
\ForEach{edge $e$ in $\edgeskyline$}{
    \ForEach{$0 \le i < |\edgeskyline(e)|$}{
        \tcc{compute active time}
        \lIf{$i=0$}{$\edgeskyline(e)[0].\yestime \gets Ts$}
        \lElse{$\edgeskyline(e)[i].\yestime \gets \edgeskyline(e)[i-1].start+1$}        
    }
}

\ForEach{$Ts\le t \le Te$}{
    $B_a[t] \gets \emptyset$\;
    $B_s[t] \gets \emptyset$\;
}
sort labels of $\edgeskyline$ in ascending order of end times\;

\ForEach{$w \in \edgeskyline$}{
    $B_a[w.\yestime].push(w)$\;
    $B_s[w.start].push(w)$\;
}

$\dll \gets$ an empty doubly linked list\;
\ForEach{$Ts \le t \le Te$}{
    \If{$t > Ts$}{
        \ForEach{$w \in B_s[t-1]$}{
            $\algdlldelete(w)$\;
        }
    }

    $h \gets$ the dummy head node for $\dll$\;
    \ForEach{$w \in B_a[t]$}{
        \While{$h.next \neq \kwnull \land h.next < w$}{
            $h \gets h.next$\;
        }
        $\algdllinsert(w,h,h.next)$\;
        $h \gets w$\;
    }
    \lIf{$B_s[t] = \emptyset$}{\textbf{continue}}

    % start enumeration
    $\algstartanchor(\dll,t)$\;

}

\proc{$\algdlldelete(w)$}{
    $w.pre.next \gets w.next$\;
    \lIf{$w.next \neq \kwnull$}{$w.next.pre \gets w.pre$}

}

\proc{$\algdllinsert(w,a,b)$}{
    $w.next \gets b$\;
    $w.pre \gets a$\;
    $a.next \gets w$\;
    \lIf{$b \neq \kwnull$}{$b.pre \gets w$}
}

\end{algorithm}

\subsection{Enumeration for All Start Times}
\label{subsec:all}

We now study how to enumerate temporal k-cores for all start times. This is achieved by maintaining the data structure $\dll_{ts}$ for different start times $ts$. Recall that we need to maintain a set of minimal core windows in ascending order of end times in \refsubsec{start}. Assume we increase the start time from $ts$ to $ts+1$. We expect to remove all minimal core windows with the start time $ts$ from the order and add all minimal core windows with the active time $ts+1$ to the order. To this end, we implement the order for each start time as a doubly linked list and present our final algorithm for temporal $k$-core enumeration in \refalg{enum}.

\refalg{enum} takes the $\edgeskyline$ computed by \refalg{ct_comp} as input. We first compute the active time of each minimal core window in lines 1-4. Recall that the minimal core windows of each edge are computed by increasing the start time in \refsubsec{vertexcoretime}. The windows returned by \refalg{ct_comp} are naturally arranged chronologically. This enables us to sequentially scan minimal core windows of each edge and derive the active time for each window. The time complexity of this process is $O(|\edgeskyline|)$.

After computing the active time, $B_a$ (line 10) and $B_s$ (line 11) are used to collect windows for each active time and start time, respectively. Lines 8--11 arrange all windows in a certain order in $B_a$ and $B_s$. The structure for each time key in $B_a$ and $B_s$ is an array, and the $push$ function adds the item to the end of the array.

Given a window $w$, we use $w.next$ and $w.pre$ to denote the next item and the previous item in the doubly linked list, respectively. $\algdlldelete$ and $\algdllinsert$ operators are used to insert an item and delete an item in the doubly linked list, respectively.
When increasing $t$, we remove all windows with the start time $t-1$ from the doubly linked list in lines 14--16 because the windows cannot exist given the new start time $t+1$. Removing those windows is safe because they are never considered in the temporal $k$-cores given the increase of the start time.
We add all windows with the active time $t$ to the order in lines 18--22. Given that all windows in $B_a[t]$ are in ascending order of end times, we search the doubly linked list in a single direction and start from the ending item $h$ for the $w$ in the previous iteration of line 17.
$B_s[t]=\emptyset$ in line 23 means no minimal core window exists with the start time $t$. No temporal $k$-core exists starting from $t$ based on \reflem{windowstart}. We invoke $\algstartanchor$ at the end to output all temporal $k$-cores with a start time of $ts$.

\begin{example}
For the temporal graph in \reffig{example} and its $\edgeskyline$ for $k = 2$ in \reftab{ecs}, we enumerate all temporal 2-cores for the query time range $[1,6]$. Starting with $ts = 1$, we insert relevant minimal core windows into $\dll_1$ (\reffig{L1}) and apply \refalg{start} to enumerate temporal 2-cores. For $ts = 2$, minimal core windows with start times less than 2, such as $(v_2,v_9)$, $(v_2,v_3)$ and $(v_3,v_9)$, are removed from $\dll_1$. This is done efficiently by linking windows in constant time. Next, windows with an activation time of 2, like $[2,6]$ for $(v_2,v_3)$, are inserted into $\dll$ based on their position, forming $\dll_2$. \refalg{start} is then used to enumerate temporal 2-cores for $ts = 2$. This process repeats for all $ts$ values until $ts = Te$.

% For the temporal graph in \reffig{example} and its $\edgeskyline$ for $k = 2$ in \reftab{ecs}, we enumerate all the temporal $2$-cores given the query time range $[1,6]$. We first sort all the minimal core windows in $\edgeskyline$ and record their position. We start with $ts = 1$. We first insert all the relevant minimal core windows into $\dll_1$ as shown in \reffig{L1}. Then we perform \refalg{start} to enumerate all temporal $2$-cores starting at $ts = 1$. After the enumeration, we increase the start time to $ts = 2$. The minimal core windows for $(v_2,v_9)$, $(v_2,v_3)$ and $(v_3,v_9)$ with times $[1,4]$ are removed from $\dll_1$ as their start time is less than 2. The removal is performed by linking the last window of $[2,3]$ for edge $(v_2,v_3)$ to the first window of $[3,5]$ for edge $(v_4,v_8)$ in constant time. Next, we identify the minimal core windows with an activation time of 2 and insert them into $\dll$. In this case, there is only one window $w = [2,6]$ for the edge $(v_2, v_3)$ with $w.active = 2$. We insert $w$ into $\dll$ between the window $w' = [2,6]$ for edge $(v_1,v_3)$ and the window $w'' = [6,7]$ for edge $(v_3,v_5)$. Since the indices of all windows are recorded, the insertion is performed by another scan of $\dll$. After the insertion, we obtain $\dll_2$. Then, we use \refalg{start} to enumerate all temporal 2-cores with start time 2. Following the same procedures, we update $\dll_{ts}$ for all start times until we reach the $ts = Te$.
\end{example}

\begin{theorem}
\label{thm:time}
The time complexity of \refalg{enum} is $O(|\rset|)$, where $|\rset|$ represents the total size of all resulting temporal $k$-cores.
\end{theorem}

\begin{proof}
The computation of active times (lines 1-4) and the initialization of the buckets (lines 5 - 11) takes $O(|\edgeskyline|)$, which is bounded by $|\rset|$. For each start time $ts$ in the main process, the deletion process (lines 14 - 16) takes $O(|B_s(ts)|)$, which is bounded by $O(|\mathcal{L}_{ts-1}|)$; the insertion process (lines 18 - 22) takes $O(|\dll_{ts}|)$; line 24 (AS-Output) takes $O(|\dll_{ts}|)$. Let $|\rset_{ts}|$ denote the total size of the set of temporal $k$-cores for a specific start time $ts$. Since $O(|\dll_{ts}|)$ is bounded by $O(|\rset_{ts}|)$, the entire main process (lines 13-24) has a time complexity of $O(|\Sigma_{ts \in [1,t_{max}]}\rset_{ts}|) = O(|\rset|)$. Therefore, the time complexity of the algorithm is $O(|\rset|)$.
\end{proof}

Based on \refthm{cttime} and \refthm{time}, the total time complexity of \refalg{ct_comp} and \refalg{enum} is $O(|\vctidx| \cdot \avgdegree+|\rset|)$.

%% file: 6experiment.tex
\section{Experiments}
\label{sec:experiments}

\begin{figure*}[htbp]
\centering
\includegraphics[scale=0.35]{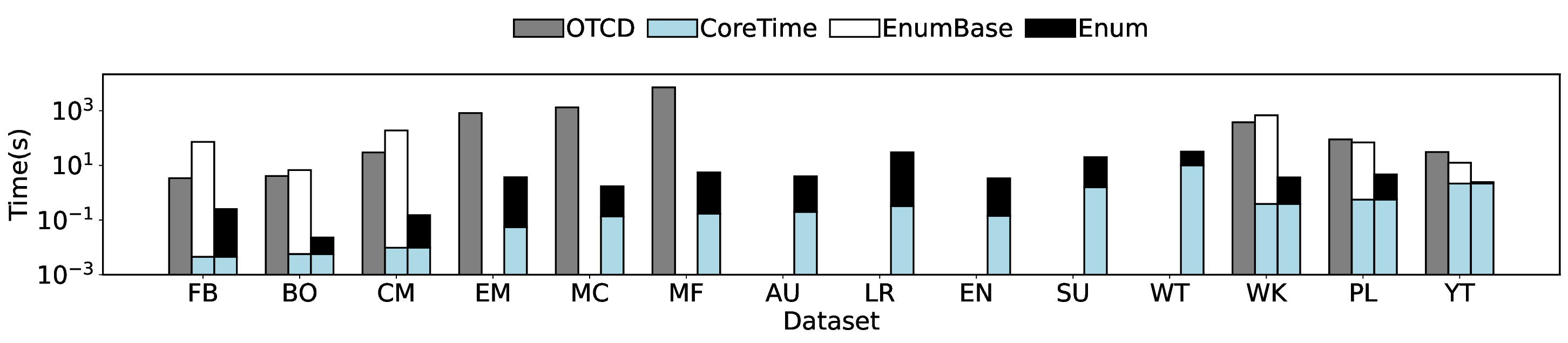}
\caption{Average running time for 30\% $k_{max}$ on query time ranges of 10\% $t_{max}$.}
\label{fig:otcd_comp}
\end{figure*}

% \begin{figure*}[tbp]
% \centering
% \includegraphics[scale=0.35]{figs/experiment/eset_all.eps}
% \caption{Running time for 30\% $k_{max}$ on time windows of 10\% $t_{max}$.}
% \label{fig:eset_all}
% \end{figure*}

\begin{figure*}[htbp]
\centering
    \subfigure[CollegeMsg]{
        \label{fig:CollegeMsg_varyk}
        \includegraphics[width=40mm]{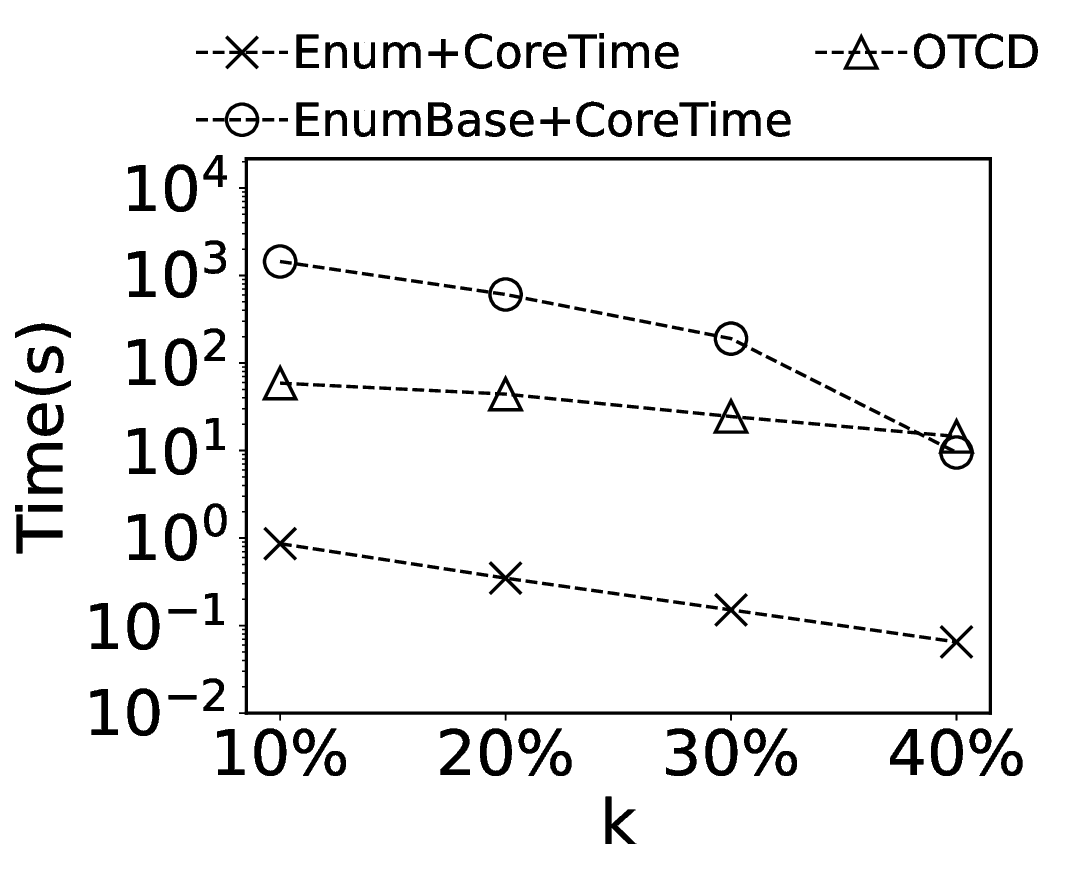}
    }
    \subfigure[Email]{
        \label{fig:email_varyk}
        \includegraphics[width=40mm]{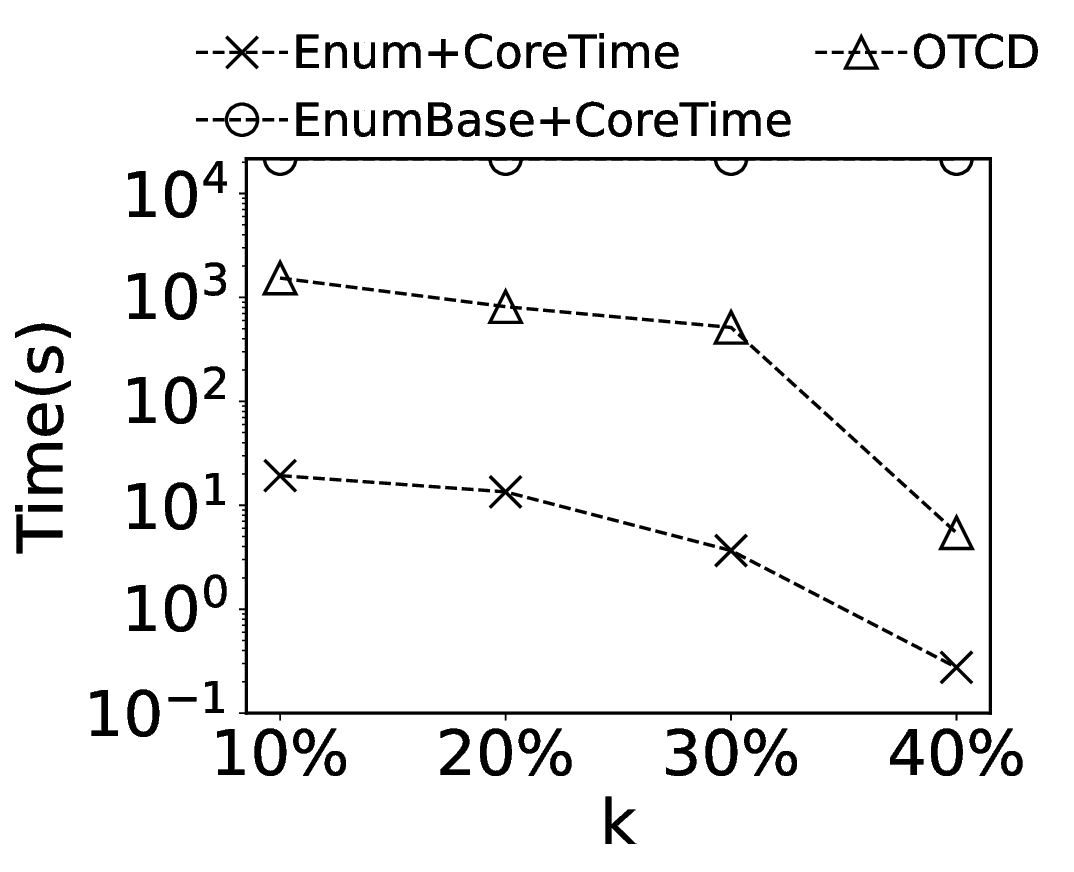}
    }
    \subfigure[Wikitalk]{
        \label{fig:wikitalk_varyk}
        \includegraphics[width=40mm]{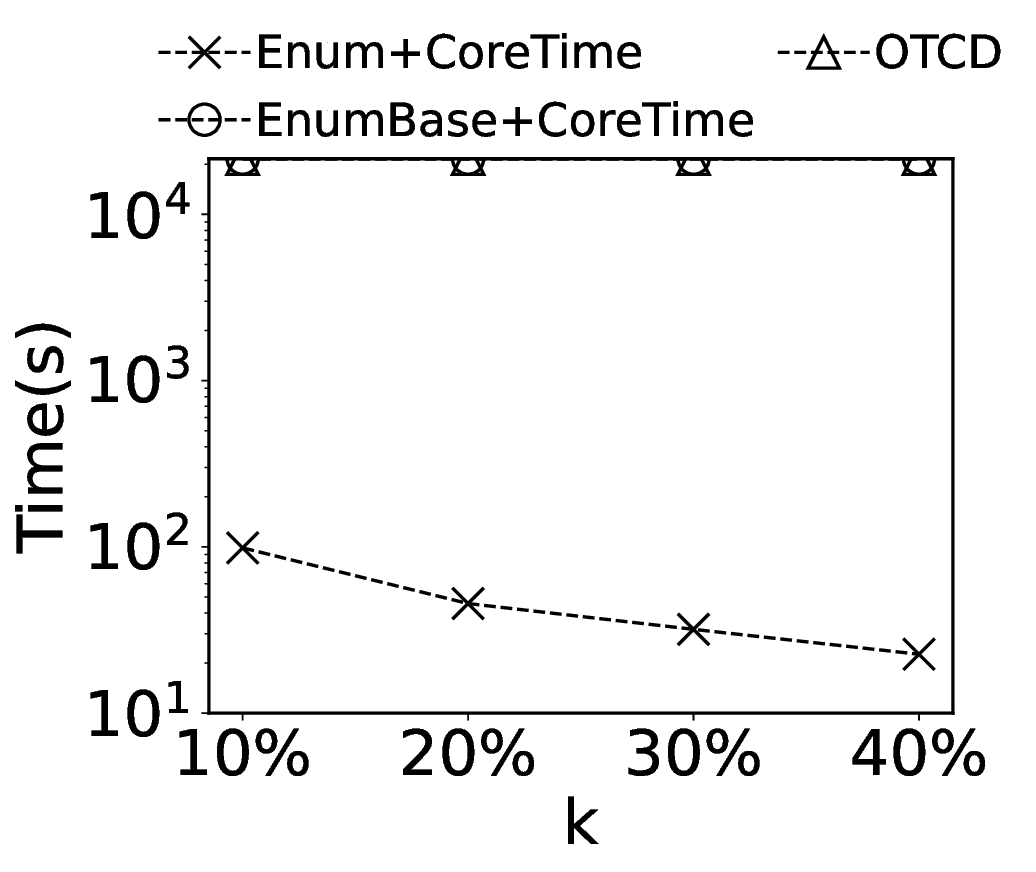}
    }
    \subfigure[Prosper]{
        \label{fig:prosper_varyk}
        \includegraphics[width=40mm]{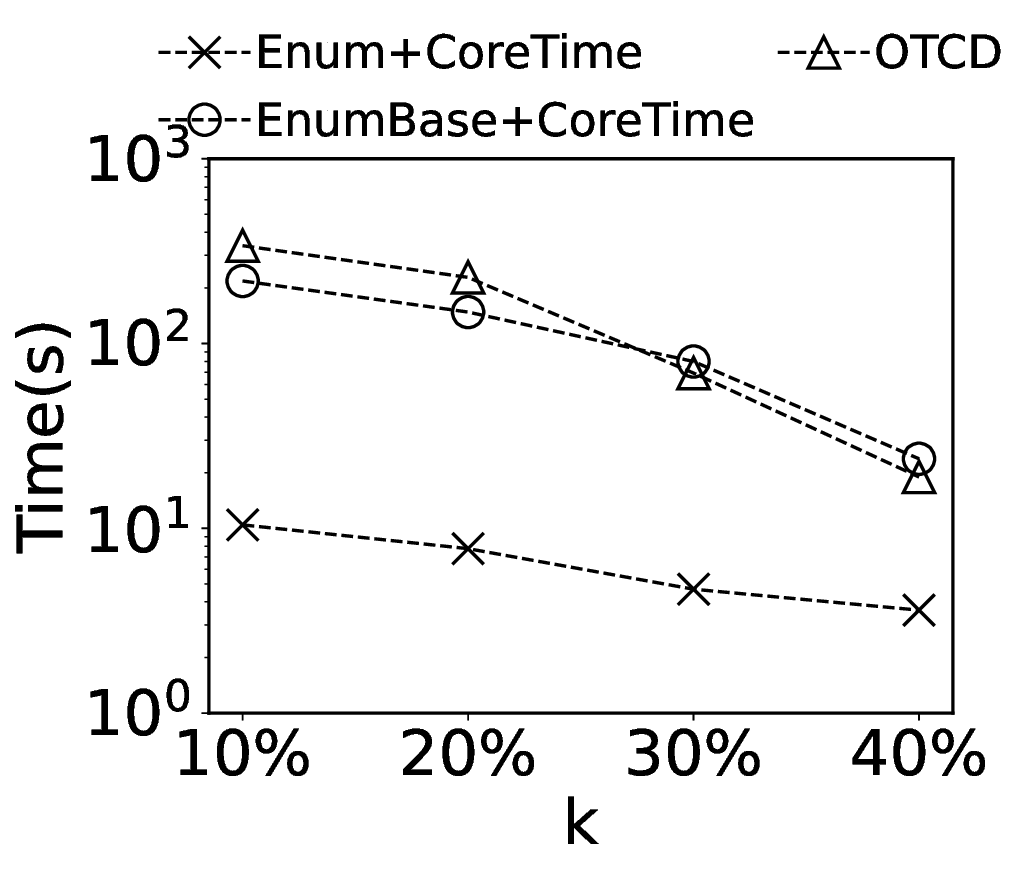}
    }
\caption{The average running time on varying $k$ between 10\%, 20\%, 30\% and 40\% of $k_{max}$.}
\label{fig:varyk}
\end{figure*}

\begin{figure*}[htbp]
\centering
    \subfigure[CollegeMsg]{
        \label{fig:CollegeMsg_varyt}
        \includegraphics[width=40mm]{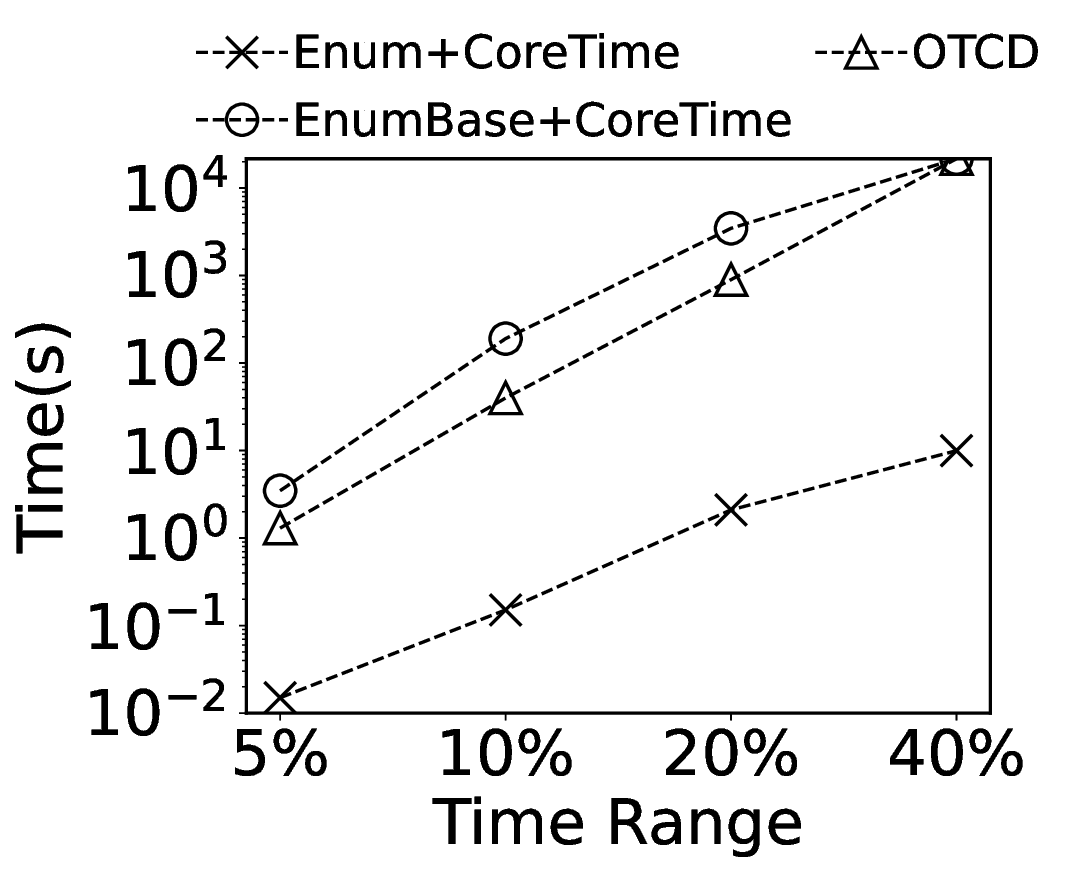}
    }
    \subfigure[Email]{
        \label{fig:email_varyt}
        \includegraphics[width=40mm]{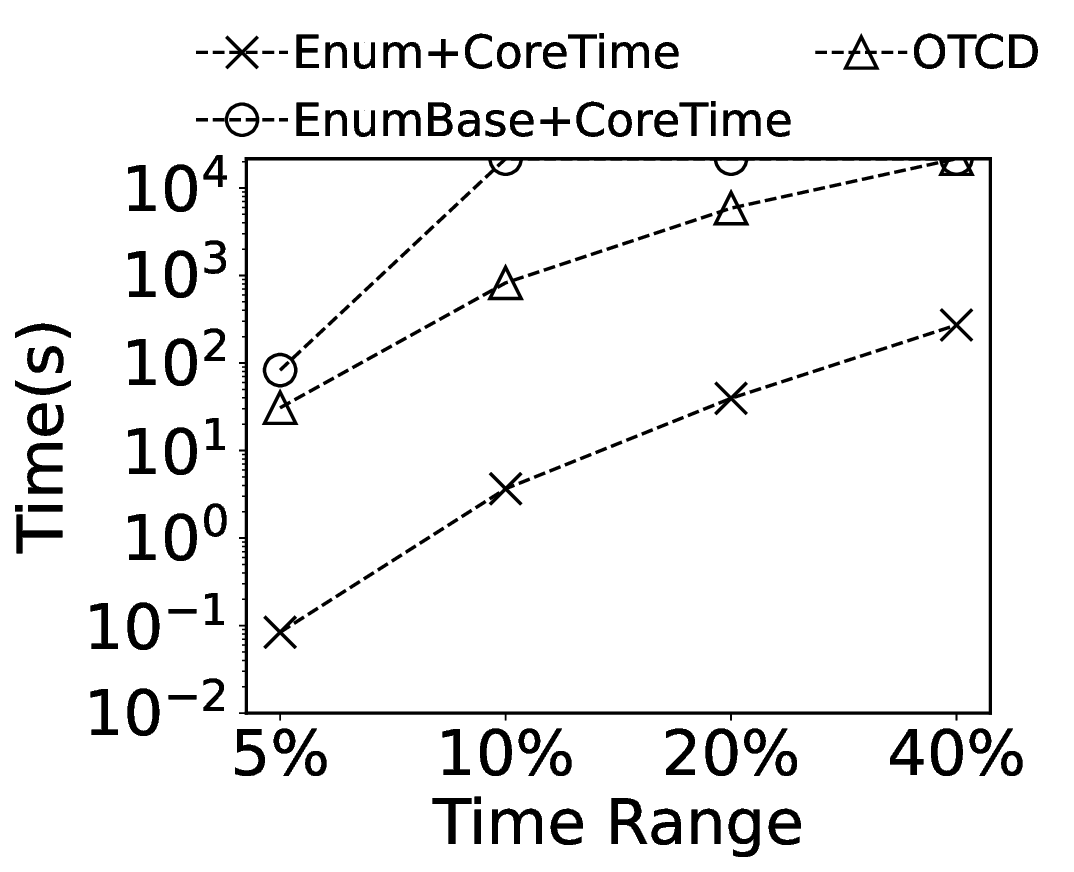}
    }
    \subfigure[Wikitalk]{
        \label{fig:wikitalk_varyt}
        \includegraphics[width=40mm]{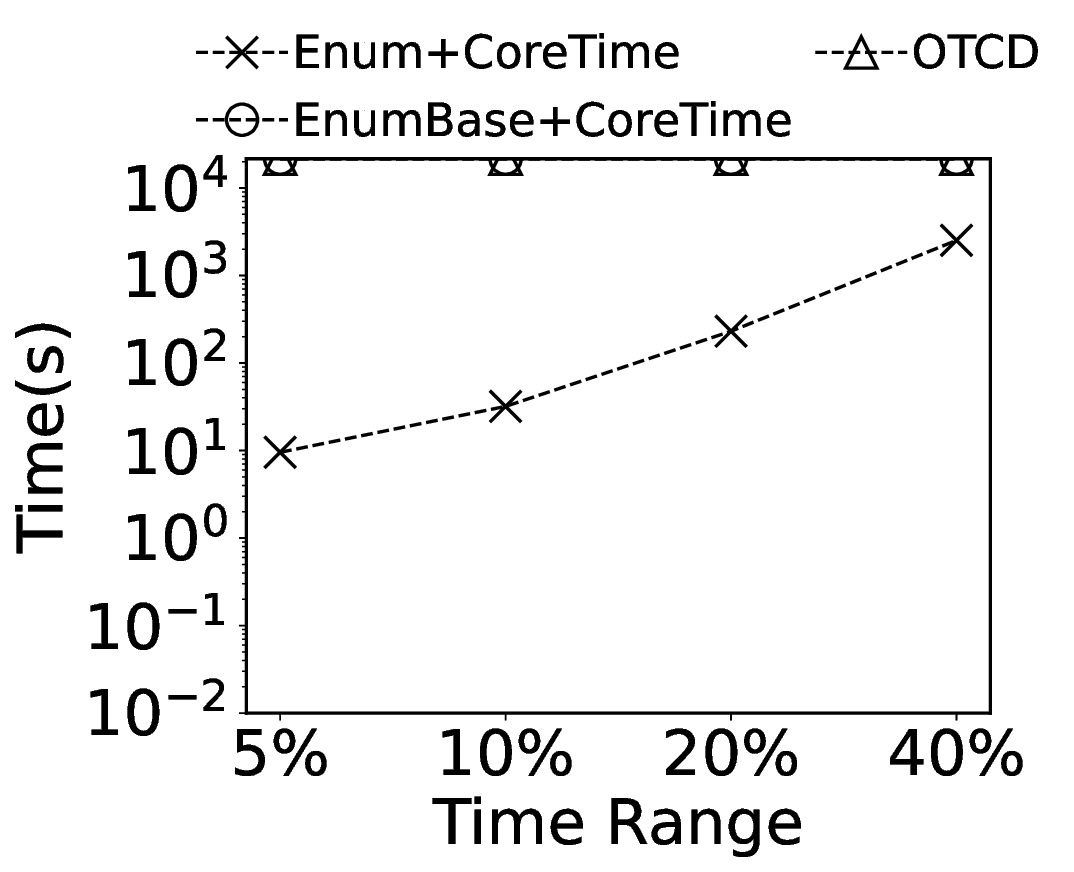}
    }
    \subfigure[Prosper]{
        \label{fig:prosper_varyt}
        \includegraphics[width=40mm]{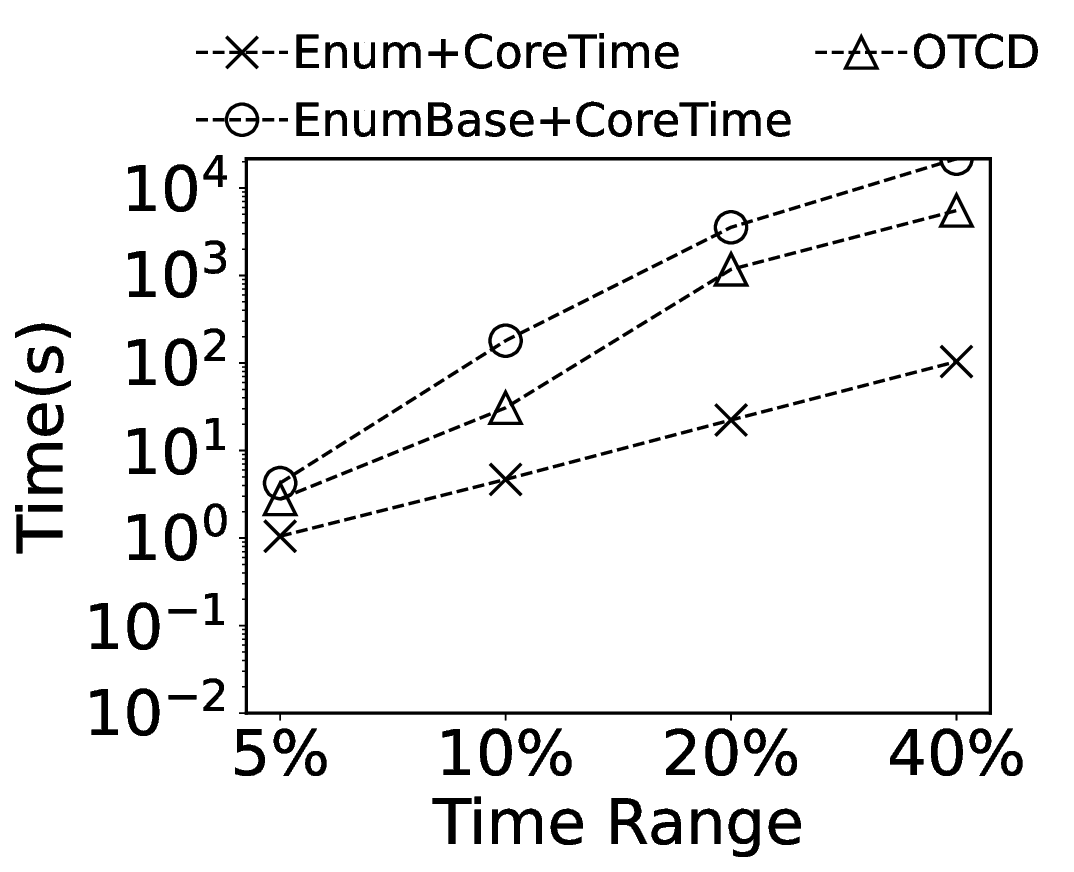}
    }
\caption{The average running time on varying query time range between 5\%, 10\%, 20\% and 40\% of $t_{max}$.}
\label{fig:varyt}
\end{figure*}

In this section, we conduct experiments to evaluate the effectiveness and efficiency of our proposed algorithm. We run all experiments on a Linux machine with Intel i9-12900K CPU and 500GB RAM. All algorithms are implemented in C++ and compiled using g++ compiler at -O3 optimization level. Implementation of OTCD is provided by the author of \cite{Yang2023}.

\stitle{Datasets.} We use fourteen real-world temporal graphs of varying sizes for our experiments. Table \ref{tab:datasets} presents their basic statistics, where $t_{max}$ is the number of distinct timestamps in the graph, and $k_{max}$ is the maximum core number among all vertices. These datasets are sourced from SNAP \cite{snap} and the KONECT Project \cite{konect}.

\stitle{Parameters.} We vary the size of the query time range to 5\%, 10\%, 20\%, and 40\% of $t_{max}$, with 10\% as the default. For the integer $k$, we vary it from 10\%, 20\%, 30\%, and 40\% of $k_{max}$, using 30\% as the default. We randomly select 100 query time ranges and record the average running time. Each query time range is guaranteed to contain at least one temporal $k$-core. We set the time limit to 6 hours, which is way beyond the maximum running time required for our final algorithm. The time points used to construct the query ranges are selected directly from the raw timestamps in the dataset, without refining or resampling their granularity. These ranges may be overlapping or disjoint, as we do not impose any artificial constraints, in order to better reflect realistic and diverse query scenarios.

\begin{table}[t]
  \caption{Datasets.}
  \label{tab:datasets}
  \centering
  \resizebox{\linewidth}{!}{  % 缩放到当前栏宽度
  \begin{tabular}{lrrrr}
    \toprule
    Name & $|V|$ & $|E|$ & $t_{max}$ & $k_{max}$\\
    \midrule
    FB-Forum (FB) & 899 & 33,786 & 33,482 & 19\\
    BitcoinOtc (BO) & 5,881 & 35,592 & 35,444 & 21\\
    CollegeMsg (CM) & 1,899 & 59,835 & 58,911 & 20\\
    Email (EM) & 986 & 332,334 & 207,880 & 34\\
    Mooc (MC) & 7,143 & 411,749 & 345,600 & 76\\
    MathOverflow (MO) & 24,818 & 506,550 & 505,784 & 78\\
    AskUbuntu (AU) & 159,316 & 964,437 & 960,866 & 48\\
    Lkml-reply (LR) & 63,399 & 1,096,440 & 881,701 & 91\\
    Enron (ER) & 87,273 & 1,148,072 & 220,364 & 53\\
    SuperUser (SU) & 194,085 & 1,443,339 & 1,437,199 & 61\\
    WikiTalk (WT) & 1,219,241 & 2,284,546 & 1,956,001 & 68\\
    Wikipedia (WK) & 91,340 & 2,435,731 & 4,518 & 117\\
    ProsperLoans (PL) & 89,269 & 3,394,979 & 1,259 & 111\\
    Youtube (YT) & 3,223,589 & 9,375,374 & 203 & 88\\
    \bottomrule
  \end{tabular}
  }
\end{table}

\begin{figure*}[tbp]
\centering
\includegraphics[scale=0.35]{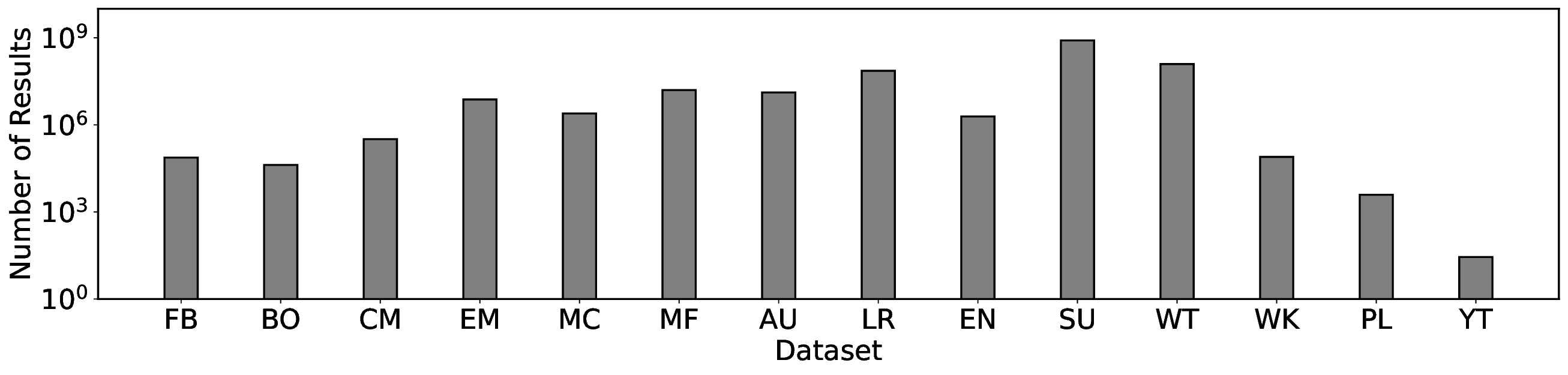}
\caption{The average number of temporal $k$-cores for each dataset run with the default parameters.}
\label{fig:num_res}
\end{figure*}

\begin{figure*}[htbp]
\centering
    \subfigure[CollegeMsg]{
        \label{fig:CollegeMsg_varyk_res}
        \includegraphics[width=40mm]{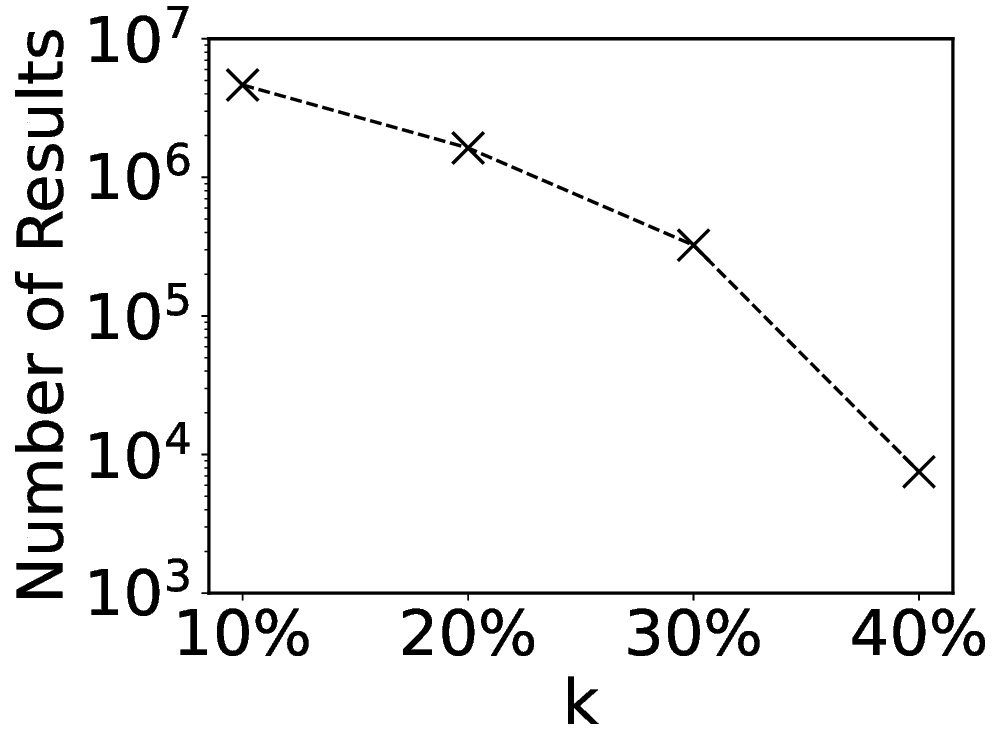}
    }
    \subfigure[Email]{
        \label{fig:email_varyk_res}
        \includegraphics[width=40mm]{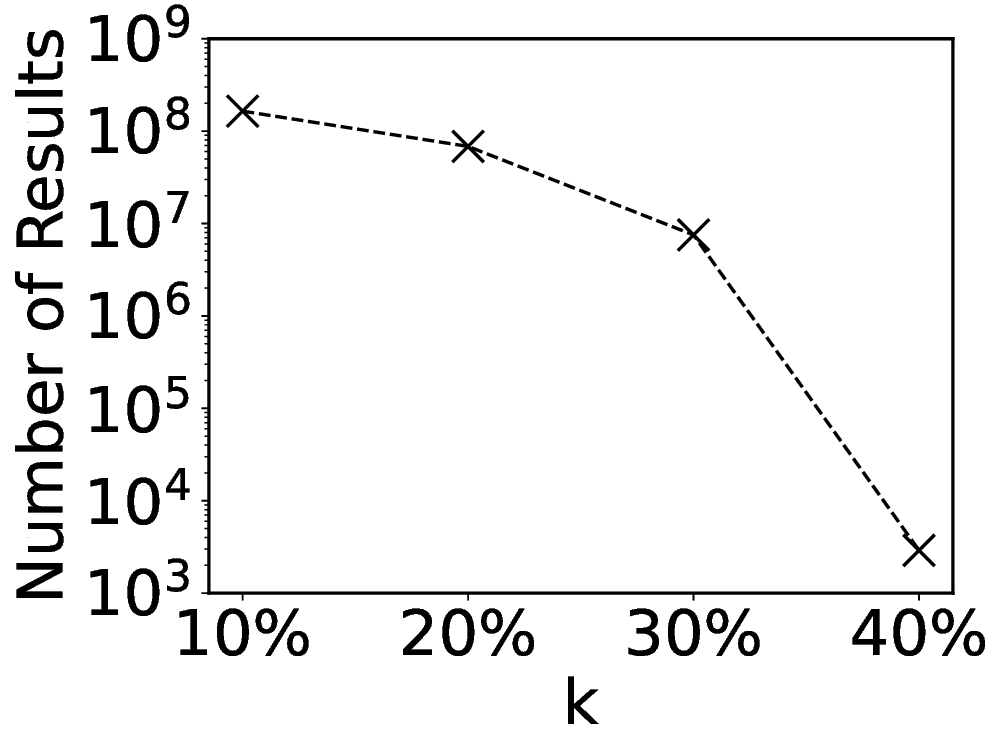}
    }
    \subfigure[Wikitalk]{
        \label{fig:wikitalk_varyk_res}
        \includegraphics[width=40mm]{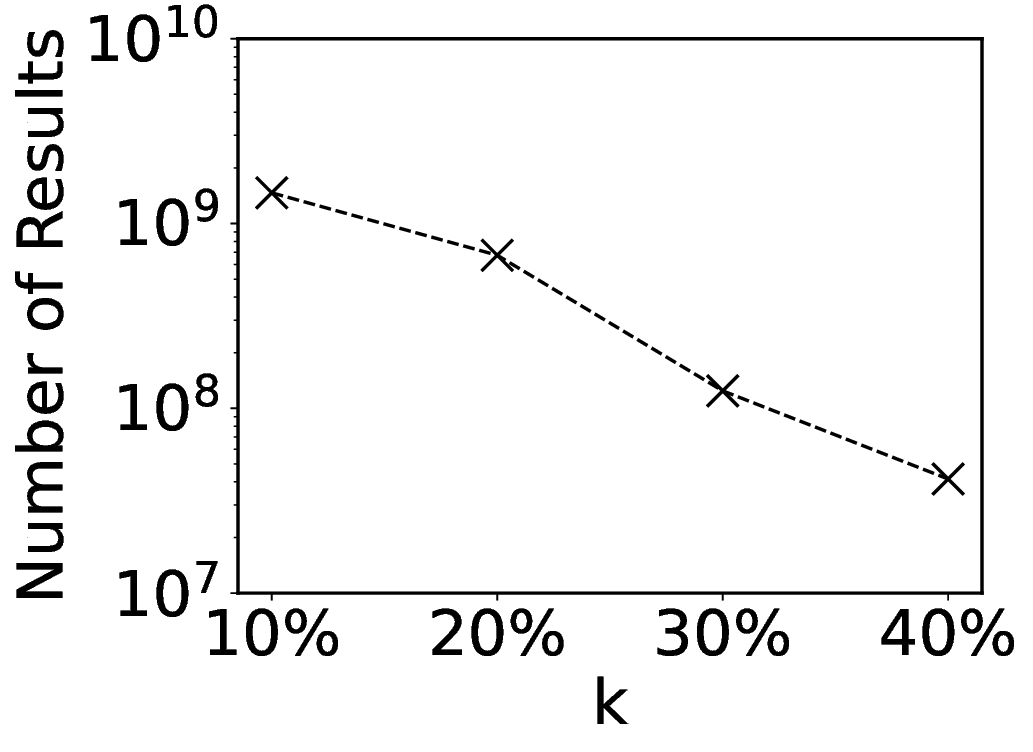}
    }
    \subfigure[Prosper]{
        \label{fig:prosper_varyk_res}
        \includegraphics[width=40mm]{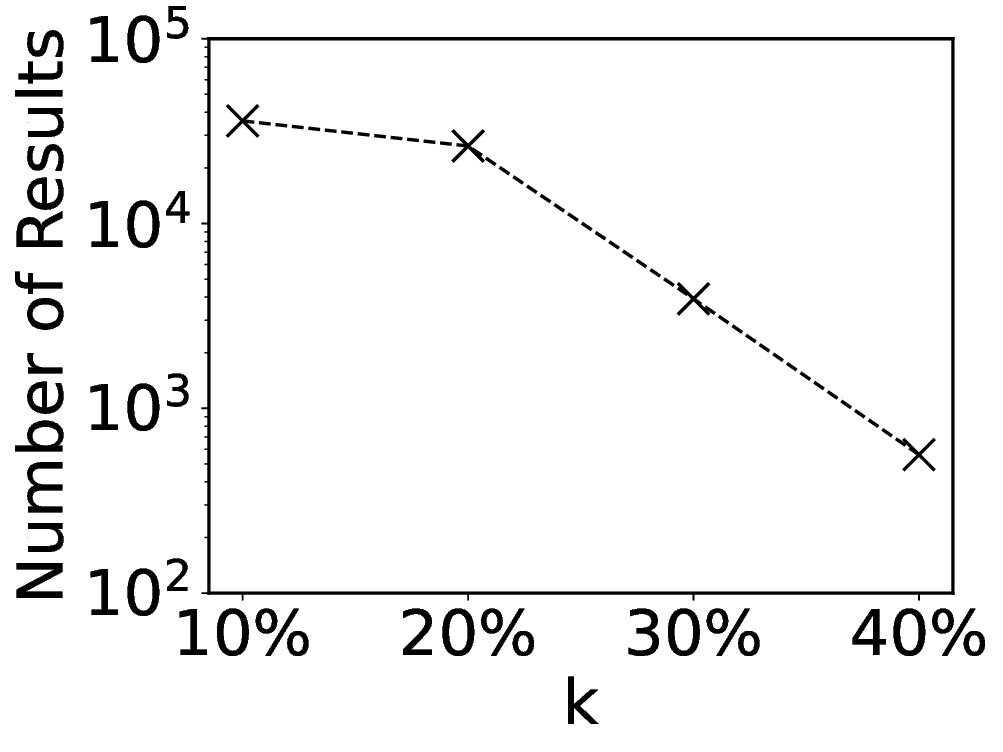}
    }
\caption{The average number of temporal $k$-cores on varying $k$ between 10\%, 20\%, 30\% and 40\% of $k_{max}$.}
\label{fig:varyk_res}
\end{figure*}

\begin{figure*}[htbp]
\centering
    \subfigure[CollegeMsg]{
        \label{fig:CollegeMsg_varyt_res}
        \includegraphics[width=40mm]{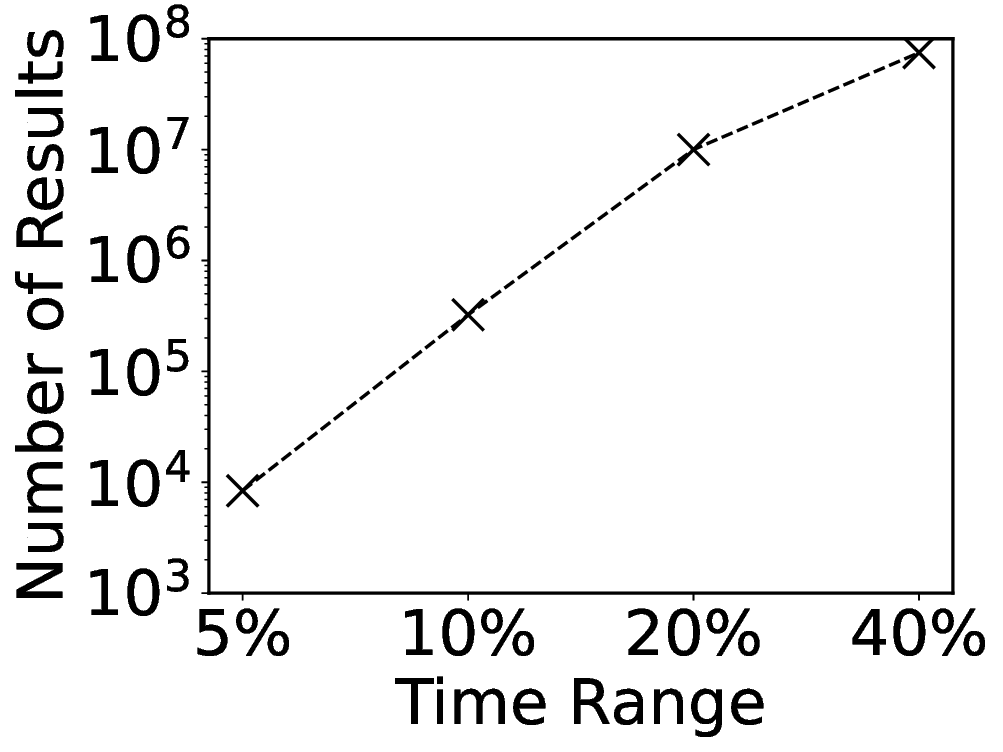}
    }
    \subfigure[Email]{
        \label{fig:email_varyt_res}
        \includegraphics[width=40mm]{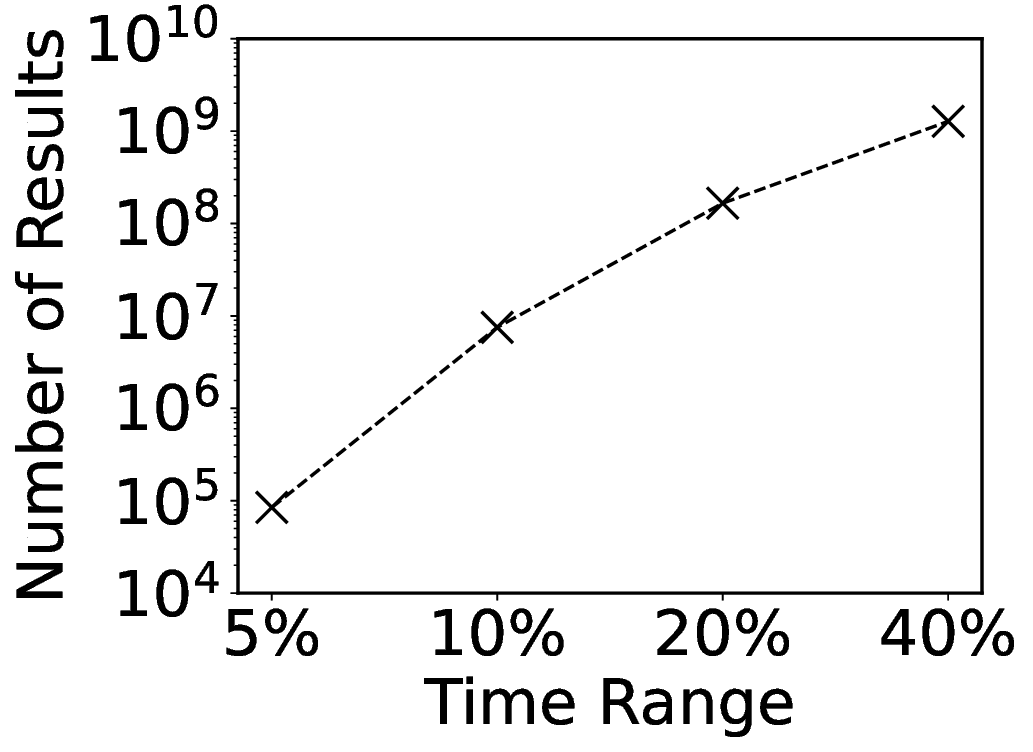}
    }
    \subfigure[Wikitalk]{
        \label{fig:wikitalk_varyt_res}
        \includegraphics[width=40mm]{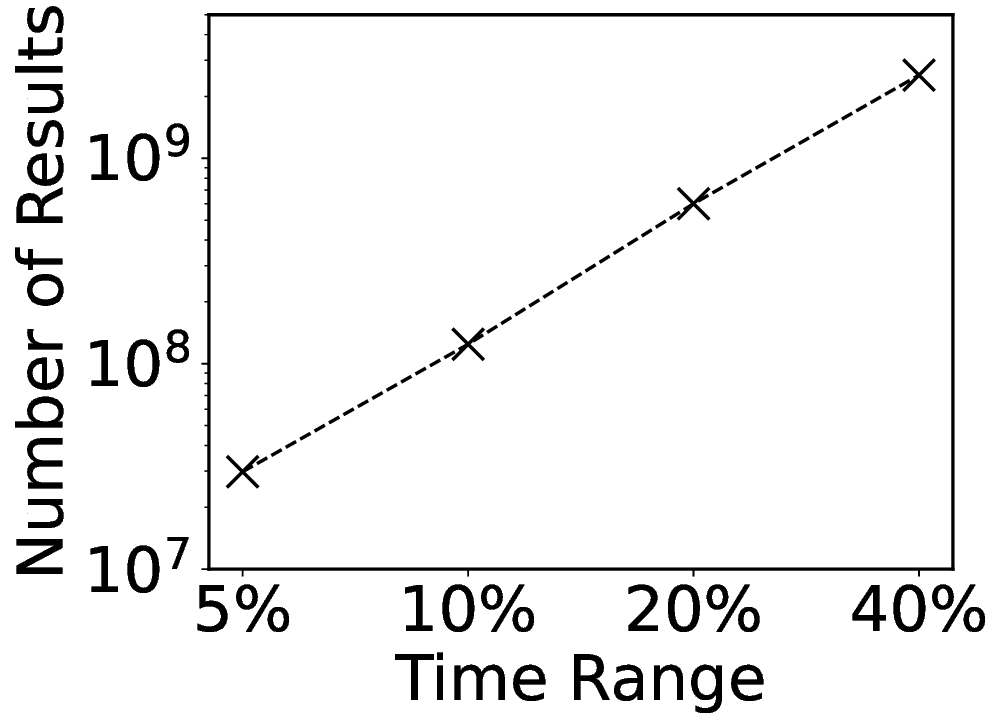}
    }
    \subfigure[Prosper]{
        \label{fig:prosper_varyt_res}
        \includegraphics[width=40mm]{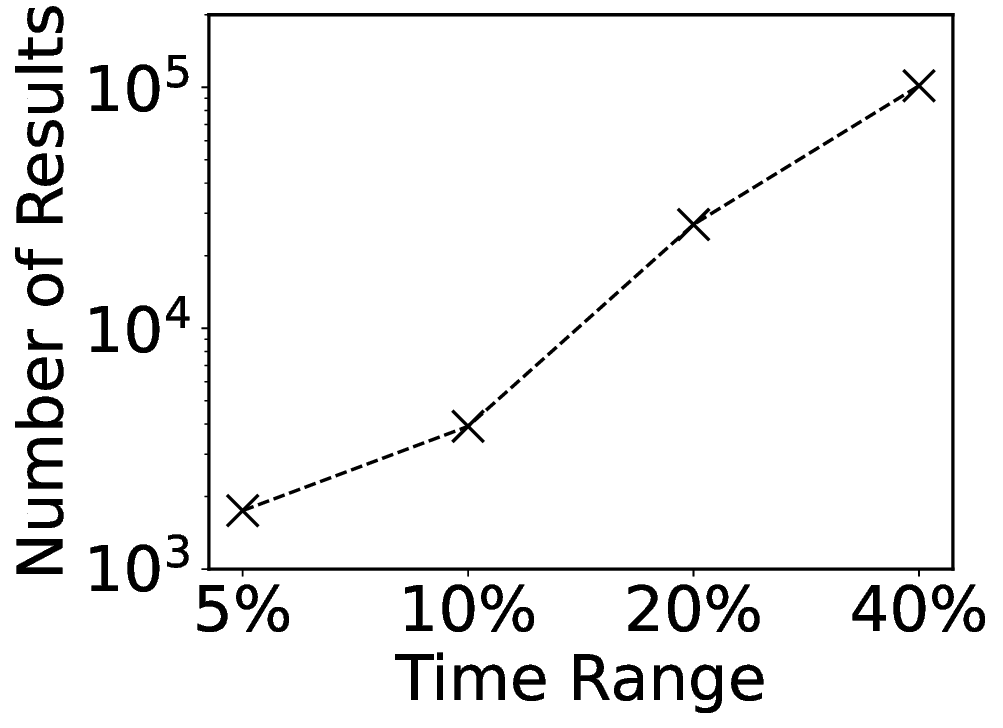}
    }
\caption{The average number of temporal $k$-cores on varying time range between 5\%, 10\%, 20\% and 40\% of $t_{max}$.}
\label{fig:varyt_res}
\end{figure*}

\subsection{Efficiency}

We compare the efficiency of our final algorithm with the baseline and the state-of-the-art OTCD algorithm, as shown in \reffig{otcd_comp}. The results demonstrate significant performance differences across all datasets. OTCD fails to complete within the time limit on five datasets with many edges and distinct timestamps, while the baseline completes only on smaller datasets or those with fewer timestamps. In contrast, our final algorithm consistently outperforms both, completing all datasets with a 2-4 order of magnitude improvement. For example, on FB, BO, and CM, it surpasses OTCD by 2-3 orders and the baseline by 3-4 orders, with even greater gains of up to 4 orders on EM, MC, and MF. Interestingly, datasets like WK, PL, and YT, despite their large edge counts, have shorter execution times due to fewer distinct timestamps. For instance, YT, with only 20 distinct timestamps at 10\% $t_{max}$, shows less drastic performance differences, but our algorithm still outperforms OTCD by an order of magnitude. This superior performance is due to our algorithm's efficient design, which eliminates redundant result verification across timestamps, streamlining computation. 
It is worth noting that the precomputation cost (computing core times, shown in blue in \reffig{otcd_comp}) is identical for the baseline and the advanced algorithm. On datasets with many distinct timestamps, the precomputation cost is relatively small, contributing only 1-10\% of the advanced algorithm's total runtime. However, for datasets with fewer timestamps, such as WK, PL, and YT, the precomputation step constitutes a larger fraction of the total runtime, reflecting the reduced time spent on enumeration in these cases.

\begin{figure*}[tbp]
\centering
 
   \includegraphics[scale=0.35]{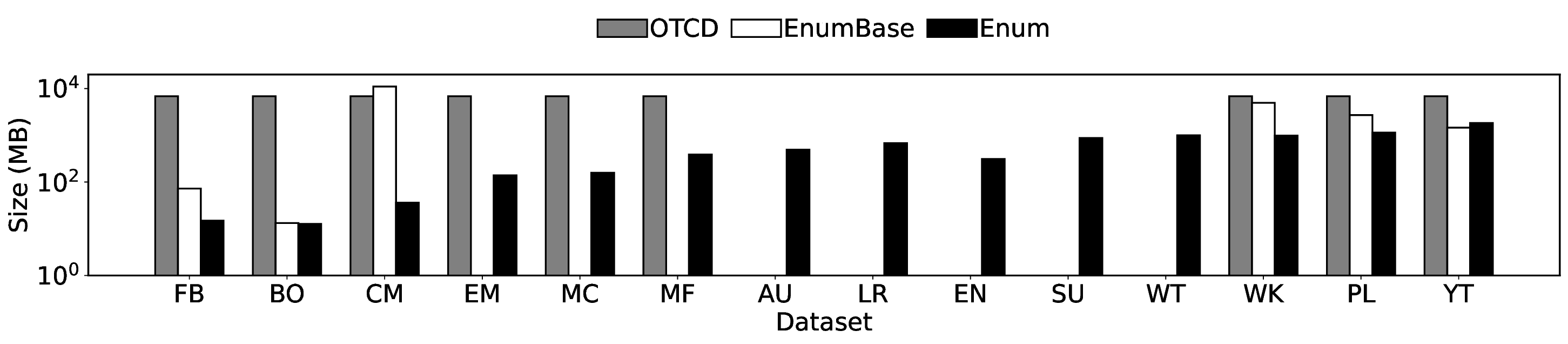}
\caption{The maximum residential memory of each algorithm run on the default parameter.}
\label{fig:memory}
\end{figure*}

% \begin{figure*}[tbp]
% \centering
% \includegraphics[scale=0.35]{figs/experiment/vset_all.eps}
% \caption{Running time of computing temporal $k$-core vertex sets.}
% \label{fig:vset_all}
% \end{figure*}

% \begin{figure*}[tbp]
% \centering
% \includegraphics[scale=0.35]{figs/experiment/res_size.eps}
% \caption{Comparison of result sets of Temporal $k$-cores and Temporal $k$-core vertex sets.}
% \label{fig:size_compare}
% \end{figure*}

\stitle {Varying $k$.} We selected four representative datasets for this experiment: CollegeMsg and Email as examples of small datasets, WikiTalk as a large dataset with a high number of timestamps, and Prosper as a large dataset with relatively few timestamps. We varied the value of $k$ from 10\% to 40\% of $k_{max}$ to assess the impact on running time, as illustrated in \reffig{varyk}. The results show that for CollegeMsg, Email, and WikiTalk, the running time consistently decreases as $k$ increases. This is because larger $k$ values result in fewer temporal $k$-cores, which reduces the computational burden. Precisely, for CollegeMsg and WikiTalk, the running time for 40\% of $k_{max}$ is approximately one-tenth of that for 10\% of $k_{max}$, while for Email, the reduction is even more significant, nearing a factor of one-hundredth. In contrast, the running time for Prosper remains relatively stable as $k$ varies. This stability is attributed to the small number of distinct timestamps in the Prosper dataset, which results in denser temporal $k$-cores that are less sensitive to changes in $k$. As $k$ increases from 10\% to 40\% of $k_{max}$, the running time decreases by only about 30\%, reflecting the dense and stable nature of the temporal $k$-cores in Prosper. The performance of OTCD and our baseline algorithm follows a similar trend, where the running time generally decreases as $k$ increases. Overall, the results indicate that our final algorithm exhibits a robust decrease in running time as $k$ increases.

% We evaluated four datasets: CollegeMsg and Email (small datasets), WikiTalk (large dataset with many timestamps), and Prosper (large dataset with fewer timestamps). $k$ was varied from 10\% to 40\% of $k_{max}$ to assess its impact on running time (\reffig{varyk}). For CollegeMsg, Email, and WikiTalk, running time decreases as $k$ increases due to fewer temporal $k$-cores. Specifically, the running time for CollegeMsg and WikiTalk at 40\% of $k_{max}$ is about one-tenth of that at 10\%, while for Email, it reduces by nearly a factor of one-hundredth. Prosper, however, shows stable running times due to its dense temporal $k$-cores, with only a 30\% decrease in running time as $k$ increases. This stability is attributed to the small number of distinct timestamps in the Prosper dataset, which results in denser temporal $k$-cores that are less sensitive to changes in $k$. OTCD and the baseline exhibit a similar trend. Overall, our final algorithm demonstrates robust performance, with running time decreasing consistently as $k$ increases.

\stitle{Varing time range.} We varied the query time ranges from 5\% to 40\% of $t_{max}$ across four selected datasets and reported the results in \reffig{varyt}. As expected, the running time increases significantly with larger time ranges due to the growth in the number of output results. The running time increased by 2 to 3 orders of magnitude across all datasets as the time range expanded from 5\% to 40\% of $t_{max}$. Notably, doubling the time range resulted in approximately a five times increase in running time for Prosper and nearly a ten times increase for WikiTalk. For Email, the increase in the time range from 5\% to 10\% of $t_{max}$ caused a drastic fifty times increase in running time; thereafter, the running time rose roughly ten times with each subsequent doubling of the time range. These observations underscore the significant impact of expanding the time window on both the number of results and the overall result size. OTCD exhibited a similar trend across all datasets but with even greater challenges. For CollegeMsg, it could not be completed within the time limit when the time range reaches 40\% of $t_{max}$. On the Email dataset, OTCD could only finish within the time limit for a 5\% time range. For WikiTalk, OTCD is unable to run for any time range due to the overwhelming number of results. Overall, these results highlight the scalability and efficiency challenges posed by larger time ranges, particularly for the baseline algorithm. In contrast, our final algorithm consistently manages to handle these increasing complexities more effectively, demonstrating its robustness and practical utility in large-scale temporal graph analysis.

\subsection{Number of Temporal $k$-cores} \reffig{num_res} shows the average number of temporal $k$-cores across datasets using default parameters. Datasets with more distinct timestamps, such as SU and WT, generate the most $k$-cores, while WK, PL, and YT, despite their large edge counts, produce fewer due to limited timestamps. \reffig{varyk_res} illustrates the effect of varying $k$ (10\%-40\% of $k_{max}$): the number of $k$-cores decreases as $k$ increases. For CollegeMsg and Email, results drop by 3-4 orders of magnitude from 30\% to 40\%, while Wikitalk and Prosper show a 2-order reduction from 10\% to 40\%. \reffig{varyt_res} highlights the impact of query time range (5\%-40\% of $t_{max}$): as the range grows, $k$-cores increase significantly, doubling by 2 orders of magnitude for CollegeMsg, Email, Wikitalk, and Prosper as the range expands.

\subsection{Memory Overhead}
\reffig{memory} compares the memory usage of each algorithm with default parameters. OTCD consistently incurs a memory cost of ~7 GB on datasets it completes, while the baseline consumes even more due to storing previously generated temporal $k$-cores for comparisons. In contrast, our final algorithm uses significantly less memory across all datasets, maintaining a footprint of under 2 GB, except for YT, which has very few distinct timestamps. This efficiency is achieved by avoiding the storage of previously generated $k$-cores or maintaining them as subgraphs, unlike OTCD. Notably, datasets like WK, PL, and YT show the highest memory overhead despite having fewer distinct $k$-cores, due to their denser cores with many edges sharing the same timestamps.

% \reffig{memory} shows the memory overhead of each algorithm run with the default parameters. The implementation of OTCD provided by the author consistently incurs a memory cost of approximately 7 GB on all the datasets which it is able to complete. In contrast, our baseline approach consumes even more memory as it stores previously generated temporal $k$-cores for future comparisons. Our final algorithm, however, demonstrates significantly lower memory usage compared to the baseline across all datasets, except for YT, which has a very small number of distinct timestamps. This reduction in memory overhead is due to the fact that our final algorithm avoids storing previously generated temporal $k$-cores and does not maintain them as subgraphs, unlike OTCD. Consequently, our final algorithm maintains a memory footprint of less than 2 GB on all datasets. Interestingly, despite having a smaller number of distinct temporal $k$-cores, datasets such as WK, PL, and YT exhibit the highest memory overhead. This is attributed to the larger size of their temporal $k$-cores, where many edges share the same timestamps, resulting in denser cores.